\documentclass[letter,11pt]{article}
\usepackage[margin=1in]{geometry}
\usepackage{amsmath,amsfonts,amssymb,amsthm,enumitem}
\usepackage{mathtools,thmtools,thm-restate}
\usepackage{color}
\usepackage{hyperref}
\usepackage{xspace}
\usepackage{subcaption}
\usepackage{framed}
\usepackage{tikz}
\usetikzlibrary{shapes.geometric,arrows,decorations.markings,positioning}

\theoremstyle{plain}
\newtheorem{open-question}{Open Question}
\newtheorem{main-theorem}{Theorem}[section]
\newtheorem{main-conjecture}[main-theorem]{Conjecture}
\newtheorem{main-lemma}[main-theorem]{Lemma}
\newtheorem{main-corollary}[main-theorem]{Corollary}
\newtheorem{main-claim}[main-theorem]{Claim}
\theoremstyle{definition}
\newtheorem{main-definition}[main-theorem]{Definition}

\newcommand{\eofasym}{\ensuremath\triangleleft}
\makeatletter
\newenvironment{algorithm-description}[2][Algorithm]
{\def\@currentlabel{#2}\gdef\eofa{\hfill\eofasym}\paragraph{#1 #2.}}
{\eofa\bigskip}
\makeatother
\newcommand\eofahere{\ifmmode{\tag*{\eofasym}}\else\eofa\fi\gdef\eofa{}}

\theoremstyle{plain}
\newtheorem{theorem}{Theorem}[subsection]
\newtheorem{thm}{Theorem}[subsection]

\newtheorem{lemma}[theorem]{Lemma}
\newtheorem{claim}[theorem]{Claim}

\newtheorem*{theorem*}{Theorem}
\newtheorem*{corollary*}{Corollary}
\newtheorem*{lemma*}{Lemma}

\newtheorem{corollary}[theorem]{Corollary}
\theoremstyle{definition}
\newtheorem{definition}[theorem]{Definition}
\newtheorem{dfn}[theorem]{Definition}
\newtheorem*{definition*}{Definition}

\providecommand{\unknown}{secret\xspace}
\providecommand{\anunknown}{a secret\xspace}

\DeclareMathOperator*{\E}{\mathbb{E}}

\providecommand{\Opt}{\mathrm{Opt}}
\providecommand{\rent}[1][]{r^H_{#1}}
\providecommand{\rhuf}[1][]{r^\Opt_{#1}}
\providecommand{\uent}[1][]{u^H_{#1}}
\providecommand{\uhuf}[1][]{u^\Opt_{#1}}
\providecommand{\cF}{\mathcal{F}}
\providecommand{\cQ}{\mathcal{Q}}
\providecommand{\qcomp}{\cQ_{\prec}}
\providecommand{\qeq}{\cQ_{=}}
\newcommand{\supp}[1]{\ensuremath{\mathrm{supp}(#1)}\xspace}
\providecommand{\sym}[1][X_n]{\mathrm{Sym}(#1)}

\providecommand{\Mid}{\mathrm{mid}}

\providecommand{\xmid}{{x_{\Mid}}}
\providecommand{\pimax}{\pi_{\max}}
\providecommand{\xmax}{x_{\max}}
\providecommand{\eps}{\epsilon}

\providecommand{\cone}[1]{\mathfrak{C}(#1)}
\providecommand{\epsparam}{\beta} 
\newcommand{\maxdensof}[1]{\ensuremath{\rho(#1)}\xspace}
\newcommand{\reldensof}[2]{\ensuremath{\rho_{#2}(#1)}\xspace}
\newcommand{\smallestdens}[1]{\ensuremath{\rho_{\min}(#1)}\xspace}
\newcommand{\smallestdenstitle}[1]{\texorpdfstring{\smallestdens{#1}}{rho\_min(#1)}}
\newcommand{\perm}{\ensuremath{\sigma}\xspace}

\newcommand{\halfcircle}{window\xspace}

\newif\ifdraft
\draftfalse

\ifdraft
\newcommand{\ariel}[1]{{\color{blue}{\textit{#1 --- ariel}}}}
\newcommand{\yuvalf}[1]{{\color{green}{\textit{#1 --- yuval filmus}}}}

\newcommand{\new}[1]{{\color{red}{\textit{#1}}}}
\newcommand{\yuvald}[1]{{\color{purple}{\textit{#1 --- yuval dagan}}}}

\else
\newcommand{\ariel}[1]{}
\newcommand{\yuvalf}[1]{}
\newcommand{\new}[1]{}
\newcommand{\yuvald}[1]{}

\fi

\newcommand{\defeq}{\ensuremath{:=}\xspace}

\newcommand{\set}[1]{\ensuremath{\{#1\}}\xspace}
\newcommand{\dist}{\ensuremath{\mu}\xspace}
\newcommand{\range}{\ensuremath{X_n}\xspace}
\newcommand{\groundset}{\ensuremath{X_n}\xspace}
\newcommand{\dyadof}[1]{\ensuremath{\mathrm{Spl}(#1)}\xspace}
\newcommand{\dyad}{\ensuremath{D}\xspace}
\newcommand{\hitter}{\ensuremath{{\cal Q}}\xspace}
\newcommand{\dyadhitter}{dyadic hitter\xspace}

\newcommand{\opt}[1]{\ensuremath{\mathrm{Opt}(#1)}\xspace} 

\newcommand{\costt}[2]{\ensuremath{c(#1,#2)}\xspace}

\title{Twenty (simple) questions}
\author{Yuval Dagan, Yuval Filmus, Ariel Gabizon, and Shay Moran}

\begin{document}

\tikzset{
itria/.style={
  draw,dashed,shape border uses incircle,
  isosceles triangle, isosceles triangle stretches=true, shape border rotate=90, inner sep=0,font=\small, yshift=-3cm}
}

\maketitle

\begin{abstract}
A basic combinatorial interpretation of Shannon's entropy function is via the ``20 questions'' game. This cooperative game is played by two players, Alice and Bob: Alice picks a distribution $\pi$ over the numbers $\{1,\ldots,n\}$, and announces it to Bob. She then chooses a number $x$ according to $\pi$, and Bob attempts to identify $x$ using as few Yes/No queries as possible, on average.

An optimal strategy for the ``20 questions'' game is given by a Huffman code for $\pi$: Bob's questions reveal the codeword for $x$ bit by bit. This strategy finds $x$ using fewer than $H(\pi)+1$ questions on average. However, the questions asked by Bob could be arbitrary. In this paper, we investigate the following question:
\emph{Are there restricted sets of questions that match the performance of Huffman codes, either exactly or approximately?}

Our first main result shows that for every distribution $\pi$, Bob has a strategy that uses only questions of the form ``$x < c$?'' and ``$x = c$?'', and uncovers $x$ using at most $H(\pi)+1$ questions on average, matching the performance of Huffman codes in this sense.
We also give a natural set of $O(rn^{1/r})$ questions that achieve a performance of at most $H(\pi)+r$, and show that $\Omega(rn^{1/r})$ questions are required to achieve such a guarantee.

Our second main result gives a set $\cQ$ of $1.25^{n+o(n)}$ questions such that for every distribution $\pi$, Bob can implement an \emph{optimal} strategy for $\pi$ using only questions from $\cQ$. We also show that $1.25^{n-o(n)}$ questions are needed, for infinitely many $n$.
If we allow a small slack of $r$ over the optimal strategy, then roughly $(rn)^{\Theta(1/r)}$ questions are necessary and sufficient.
\end{abstract}

\section{Introduction} \label{sec:introduction}

A basic combinatorial and operational interpretation of Shannon's entropy function, which is often taught in introductory courses on information theory, is via the ``20 questions'' game (see for example the well-known textbook~\cite{DBLP:books/daglib/0016881}).
This game is played between two players, Alice and Bob: Alice picks a distribution $\pi$ over a (finite) set of objects $X$, and announces it to Bob. Alice then chooses an object $x$ according to $\pi$, and Bob attempts to identify the object using as few Yes/No queries as possible, on average.
What questions should Bob ask?
An optimal strategy for Bob is to compute a Huffman code for~$\pi$, and then
follow the corresponding decision tree: his first query, for example, asks whether $x$ lies in the left subtree of the root. 
While this strategy minimizes the expected number of queries, the queries themselves could be arbitrarily complex;
already for the first question, Huffman's algorithm draws from an exponentially large reservoir of potential queries (see Theorem~\ref{thm:minimum-redundancy-lb} for more details).

Therefore, it is natural to consider variants of this game in which
the set of queries used is restricted;
for example, it is plausible that Alice and Bob would prefer to use queries
that (i) can be communicated efficiently (using as few bits as possible), 
and (ii) can be tested efficiently (i.e.\ there is an efficient encoding scheme for elements of $X$ and a fast algorithm that given $x\in X$ and a query $q$ as input, determines whether $x$ satisfies the query $q$).

We summarize this with the following meta-question, which guides this work: \emph{Are there ``nice''  sets of queries $\cQ$ such that for any distribution, there is a ``high quality'' strategy that uses only queries from $\cQ$?}

Formalizing this question depends on how ``nice'' and ``high quality'' are quantified.
We consider two different benchmarks for sets of queries:

\begin{enumerate}
	\item \textbf{An information-theoretical benchmark}: A set of queries $\cQ$ has \emph{redundancy $r$} if for every distribution $\pi$ there is a strategy using only queries from $\cQ$ that finds $x$ with at most $H(\pi)+r$ queries on average when $x$ is drawn according to $\pi$.
	\item \textbf{A combinatorial benchmark}: A set of queries $\cQ$ is \emph{$r$-optimal} (or has \emph{prolixity $r$}) if for every distribution $\pi$ there is a strategy using queries from $\cQ$ that finds $x$ with at most $\opt{\pi}+r$ queries on average when $x$ is drawn according to $\pi$, where $\opt{\pi}$ is the expected number of queries asked by an optimal strategy for $\pi$ (e.g.\ a Huffman tree).
\end{enumerate}

Given a certain redundancy or prolixity, we will be interested in sets of questions achieving that performance that (i) are as small as possible, and (ii) allow efficient construction of high quality strategies which achieve the target performance.
In some cases we will settle for only one of these properties, and leave the other as an open question.

\paragraph{Information-theoretical benchmark.}

Let $\pi$ be a distribution over $X$.
A basic result in information theory is that every algorithm that reveals an unknown element $x$ drawn according to $\pi$ (in short, $x \sim \pi$)
using Yes/No questions must make at least $H(\pi)$ queries on average.
Moreover, there are algorithms that make at most $H(\pi) + 1$ 
queries on average, such as Huffman coding and Shannon--Fano coding.
However, these algorithms may potentially use arbitrary queries.

Are there restricted sets of queries that match the performance of $H(\pi)+1$ queries on average, for {\bf every} distribution $\pi$?
Consider the setting in which $X$ is linearly ordered (say $X=[n]$, with its natural ordering: $1<\cdots<n$).
Gilbert and Moore~\cite{GilbertMoore}, in a result that paved the way to arithmetic coding, showed that two-way comparison queries (``$x<c$?'') almost fit the bill: they achieve a performance of at most $H(\pi)+2$ queries on average.
Our first main result shows that the optimal performance of $H(\pi)+1$ can be achieved by allowing also equality queries (``$x=c?$''):

\begin{theorem*}[restatement of Theorem~\ref{thm:redundancy-1}]
For every distribution $\pi$ there is a strategy that uses only comparison and equality queries which finds $x$ drawn from $\pi$ with at most $H(\pi)+1$ queries on average.
Moreover, this strategy can be computed in time $O(n\log n)$.
\end{theorem*}

In a sense, this result gives an affirmative answer to our main question above.
The set of comparison and equality queries (first suggested by Spuler~\cite{Spuler}) arguably qualifies as ``nice'':
linearly ordered universes appear in many natural and practical settings (numbers, dates, names, IDs)
in which comparison and equality queries can be implemented efficiently.
Moreover, from a communication complexity perspective, Bob can communicate a comparison/equality query
using just $\log_2 n + 1$ bits (since there are just $2n$ such queries). This is an exponential improvement over the $\Omega(n)$ bits he would need to communicate had he used Huffman coding.

We extend this result to the case where $X$ is a set of vectors of length $r$, $\vec x = (x_1,\ldots,x_r)$, by showing that there is a strategy using entry-wise comparisons (``$x_i<c?$'') and entry-wise equalities  (``$x_i=c?$'') that achieves redundancy $r$: 

\begin{theorem*}[restatement of Theorem~\ref{thm:redundancy-r-bounds}]
Let $X$ be the set of vectors of length $r$ over a linearly ordered universe.
For every distribution $\pi$ there is a strategy that uses only entry-wise comparison queries and entry-wise equality queries and finds $\vec x\sim\pi$ with at most $H(\pi) + r$ queries. 
\end{theorem*}

\providecommand{\find}{\mathit{find}}

The theorem is proved by applying the algorithm of the preceding theorem to uncover the vector $\vec{x}$ entry by entry. As a corollary, we are able to determine almost exactly the minimum size of a set of queries that achieves redundancy $r\geq 1$.
In more detail, let $\uent(n,r)$ denote the minimum size of a set of queries $\cQ$ such that for every distribution $\pi$ on $[n]$
there is a strategy using only queries from $\cQ$ that finds $x$ with at most $H(\pi)+r$ queries on average, when $x$ is drawn according to $\pi$.

\begin{corollary*}[Theorem~\ref{thm:redundancy-r-bounds}]
For every $n,r\in\mathbb{N}$,
\[
 \frac{1}{e} r  n^{1/ r } \leq \uent(n,r) \leq 2  r  n^{1/ r }.
\]
\end{corollary*}
Obtaining this tight estimate of $\uent(n,r) = \Theta(r  n^{1/ r })$ hinges on adding equality queries;
had we used only entry-wise comparison queries and the Gilbert--Moore algorithm instead, the resulting upper bound would have been $\uent(n,r) = O\bigl(rn^{2/r}\bigr)$, which is quadratically worse than the truth.

\medskip

\paragraph{Combinatorial benchmark.}

The analytical properties of the entropy function make $H(\pi)$ a standard benchmark for the average number of bits needed to describe an element $x$ drawn from a known distribution $\pi$, and so it is natural to use it as a benchmark for the average number of queries needed to find $x$ when it is drawn according to $\pi$.
However, there is a conceptually simpler, and arguably more natural, benchmark:
$\opt{\pi}$ --- the average number of queries that are used by a best strategy for $\pi$ (several might exist), such as one generated by Huffman's algorithm.

Can the optimal performance of Huffman codes be matched \emph{exactly}? 
Can it be achieved without using all possible queries? Our second main result answers this in the affirmative:
\begin{theorem*}[restatement of Theorem~\ref{thm:minimum-redundancy-ub} and Theorem~\ref{thm:minimum-redundancy-lb}]
For every $n$ there is a set $\cQ$ of $1.25^{n+o(n)}$ queries such that for every distribution over $[n]$, there is a strategy using only queries from $\cQ$ which matches the performance of the optimal (unrestricted) strategy \emph{exactly}.
Furthermore, for infinitely many $n$, at least $1.25^{n-o(n)}$ queries are required to achieve this feat.
\end{theorem*}

One drawback of our construction is that it is randomized.
Thus, we do not consider it particularly ``efficient''  nor ``natural''. It is interesting to find an explicit set $\cQ$ that achieves this bound. Our best explicit construction is:
\begin{theorem*}[restatement of Theorem~\ref{thm:cone}]
For every $n$ there is an explicit set $\cQ$ of $O(2^{n/2})$ queries such that for every distribution over $[n]$, there is a strategy using only queries from $\cQ$ which matches the performance of the optimal (unrestricted) strategy \emph{exactly}.
Moreover, we can compute this strategy in time $O(n^2)$.
\end{theorem*}

It is natural to ask in this setting how small can a set of queries be if it is \emph{$r$-optimal}; that is, 
if it uses at most $\opt{\pi}+r$ questions on average when the secret element is drawn according to $\pi$, for small $r > 0$.
Let $\uhuf(n,r)$ denote the minimum size of a set of queries $\cQ$ such that for every distribution $\pi$ on $[n]$
there is a strategy using only queries from $\cQ$ that finds $x$ with at most $\opt{\pi}+r$ queries on average when $x$ is drawn from $\pi$.
We show that for any fixed $r>0$, significant savings can be achieved:
\begin{theorem*}[restatement of Theorem~\ref{thm:min-redundancy-approx}]
For all $r\in(0,1)$:
\[(r\cdot n)^{\frac{1}{4r}}\lesssim \uhuf(n,r) \lesssim (r\cdot n)^{\frac{16}{r}}. \]
\end{theorem*}

Instead of the exponential number of questions needed to match Huffman's algorithm exactly, for fixed $r>0$ an $r$-optimal set of questions has polynomial size.
In this case the upper bound is achieved by an explicit set of queries $\cQ_r$.
We also present an efficient randomized algorithm for computing an $r$-optimal strategy that uses queries from $\cQ_r$.

\paragraph{Related work}
The ``20 questions'' game is the starting point of combinatorial search theory~\cite{Katona,AhlswedeWegener,ACD}. Combinatorial search theory considers many different variants of the game, such as several unknown elements, non-adaptive queries, non-binary queries, and a non-truthful Alice~\cite{RMKWS,SW,AD,DGW}. Both average-case and worst-case complexity measures are of interest.
An important topic in combinatorial search theory is combinatorial group testing~\cite{Dorfman,DuHwang}.

We are unaware of any prior work which has considered the quantities $\uent(n,r)$ or $\uhuf(n,r)$. However, several particular sets of questions have been analyzed in the literature from the perspective of redundancy and prolixity: Huffman codes~\cite{Gallager,Johnsen,CGT,MontgomeryAbrahams,Manstetten,MPK}, binary search trees~\cite{GilbertMoore,Nakatsu,Rissanen,Horibe}, and comparison-based sorting algorithms~\cite{F76,MY13}.

\section{Preliminaries} \label{sec:definitions}

\paragraph{Notation} We use $\log n$ for the base~$2$ logarithm of $n$ and $[n]$ to denote the set $\{1,\ldots,n\}$.

Throughout the document, we will consider probability distributions over the set $X_n = \{x_1,\ldots,x_n\}$ of size $n$. In some cases, we will think of this set as ordered: $x_1 \prec \cdots \prec x_n$.

If $\pi$ is a probability distribution over $X_n$, we will denote the probability of $x_i$ by $\pi_i$, and the probability of a set $S \subseteq X_n$ by $\pi(S)$.

\paragraph{Information theory} We use $H(\pi)$ to denote the base~$2$ entropy of a distribution $\pi$ and $D(\pi\|\mu)$ to denote the Kullback--Leibler divergence. The \emph{binary entropy} function $h(p)$ is the entropy of a Bernoulli random variable with success probability $p$. When $Y$ is a Bernoulli random variable, the chain rule takes the following form:
\[
 H(X,Y) = h(\Pr[Y=1]) + \Pr[Y=0] H(X|Y=0) + \Pr[Y=1] H(X|Y=1).
\]
We call this the \emph{Bernoulli chain rule}.

\paragraph{Decision trees} In this paper we consider the task of revealing \anunknown element $x$ from $X_n$ by using Yes/No questions. Such a strategy will be called a \emph{decision tree} or an \emph{algorithm}.

A decision tree is a binary tree in which the internal nodes are labeled by subsets of $X_n$ (which we call \emph{questions} or \emph{queries}), each internal node has two outgoing edges labeled \emph{Yes} (belongs to the question set) and \emph{No} (doesn't belong to the question set), and the leaves are labeled by distinct elements of $X_n$. The depth of an element $x_i$ in a decision tree $T$, denoted $T(x_i)$, is the number of edges in the unique path from the root to the unique leaf labeled $x_i$ (if any).

Decision trees can be thought of as annotated prefix codes: the code of an element $x_i$ is the concatenation of the labels of the edges leading to it. The mapping can also be used in the other direction: each binary prefix code of cardinality~$n$ corresponds to a unique decision tree over $X_n$.

Given a set $\cQ \subseteq 2^{X_n}$ (a set of \emph{allowed questions}), a \emph{decision tree using $\cQ$} is one in which all questions belong to $\cQ$. A decision tree is \emph{valid} for a distribution $\mu$ if its leaves are labeled by elements of the support of $\mu$, and each element in the support appears as the label of some leaf.

Given a distribution $\mu$ and a decision tree $T$ valid for $\mu$, the \emph{cost} (or \emph{query complexity}) of $T$ on $\mu$, labeled $T(\mu)$, is the average number of questions asked on a random element, $T(\mu) = \sum_{i=1}^n \mu_i T(x_i)$.
Given a set $\cQ$ of allowed questions and a distribution $\mu$, the \emph{optimal cost} of $\mu$ with respect to $\cQ$, denoted $\costt{\cQ}{\mu}$, is the minimal cost of a valid decision tree for $\mu$ using~$\cQ$.

\paragraph{Dyadic distributions and Huffman's algorithm}
Huffman's algorithm~\cite{Huffman} finds the optimal cost of an unrestricted decision tree for a given distribution:
\[
 \opt{\mu} = \costt{2^{X_n}}{\mu}.
\]
We call a decision tree with this cost a \emph{Huffman tree} or an \emph{optimal decision tree} for $\mu$. More generally, a decision tree is \emph{$r$-optimal} for $\mu$ if its cost is at most $\opt{\mu} + r$.

It will be useful to consider this definition from a different point of view. Say that a distribution is \emph{dyadic} if the probability of every element is either $0$ or of the form $2^{-d}$ for some integer~$d$. We can associate with each decision tree $T$ a distribution $\tau$ on the leaves of $T$ given by $\tau_i = 2^{-T(x_i)}$. This gives a correspondence between decision trees and dyadic distributions.

In the language of dyadic distributions, Huffman's algorithm solves the following optimization problem:
\[
 \opt{\mu} = \min_{\substack{\tau \text{ dyadic} \\ \supp \tau = \supp \mu}} \sum_{i=1}^n \mu_i \log \frac{1}{\tau_i} = \min_{\substack{\tau \text{ dyadic} \\ \supp \tau = \supp \mu}} \left[ H(\mu) + D(\mu \| \tau) \right].
\]
In other words, computing $\opt{\mu}$ amounts to minimizing $D(\mu\|\tau)$, and thus to ``rounding'' $\mu$ to a dyadic distribution. We call $\tau$ a \emph{Huffman distribution} for $\mu$.

The following classical inequality shows that $\opt{\mu}$ is very close to the entropy of $\mu$:
\[
 H(\mu) \leq \opt{\mu} < H(\mu) + 1.	
\]
The lower bound follows from the non-negativity of the Kullback--Leibler divergence; it is tight exactly when $\mu$ is dyadic.
The upper bound from the Shannon--Fano code, which corresponds to the dyadic sub-distribution $\tau_i = 2^{-\lceil \log \mu_i \rceil}$ (in which the probabilities could sum to less than~$1$).

\paragraph{Redundancy and prolixity}
We measure the quality of sets of questions by comparing the cost of decision trees using them to the entropy (the difference is known as \emph{redundancy}) and to the cost of optimal decision trees (for the difference we coin the term \emph{prolixity}).
In more detail, the \emph{redundancy} of a decision tree $T$ for a distribution $\mu$ is $T(\mu) - H(\mu)$, and its \emph{prolixity} is $T(\mu) - \opt{\mu}$.

Given a set of questions $\cQ$, the redundancy $\rent(\cQ,\mu)$ and prolixity $\rhuf(\cQ,\mu)$ of a distribution $\mu$ are the best redundancy and prolixity achievable using questions from $Q$:
\begin{align*}
 \rent(\cQ,\mu) &= \costt{\cQ}{\mu} - H(\mu), \\
 \rhuf(\cQ,\mu) &= \costt{\cQ}{\mu} - \opt{\mu}.
\end{align*}
The redundancy of a set of questions $\cQ$, denoted $\rent(\cQ)$, is  the supremum of $\rent(\cQ,\mu)$ over all distributions $\mu$ over $X_n$. The prolixity $\rhuf(\cQ)$ of $\cQ$ is defined similarly.
These quantities are closely related, as the inequality $H(\mu) \leq \opt{\mu} < H(\mu)+1$ implies:
\[
 \rhuf(\cQ) \leq \rent(\cQ) \leq \rhuf(\cQ) + 1.
\]

A set of questions $\cQ$ is \emph{optimal} if $\rhuf(\cQ) = 0$, and $r$-optimal if $\rhuf(\cQ) \leq r$. \yuvalf{Not sure we use the latter.}

\paragraph{The parameters $\uent(n,r)$ and $\uhuf(n,r)$}
Our main object of study in this paper are the parameters $\uent(n,r)$ and $\uhuf(n,r)$.
The parameter $\uent(n,r)$ is the cardinality of the minimal set of questions $\cQ \subseteq 2^{X_n}$ such that $\rent(\cQ) \leq r$. Similarly, the parameter $\uhuf(n,r)$ is the cardinality of the minimal set of questions $\cQ$ such that $\rhuf(\cQ) \leq r$.
These quantities are closely related:
\[
 \uhuf(n,r) \leq \uent(n,r) \leq \uhuf(n,r-1).
\]

\paragraph{A useful lemma}
The following simple lemma will be used several times in the rest of the paper.

\begin{main-lemma} \label{lem:neat-sum}
 Let $p_1 \geq \ldots \geq p_n$ be a non-increasing list of numbers of the form $p_i = 2^{-a_i}$ (for integer $a_i$), and let $a \leq a_1$ be an integer. If $\sum_{i=1}^n p_i \geq 2^{-a}$ then for some $m$ we have $\sum_{i=1}^m p_i = 2^{-a}$.
 
 If furthermore $\sum_{i=1}^n p_i$ is a multiple of $2^{-a}$ then for some $\ell$ we have $\sum_{i=\ell}^n p_i = 2^{-a}$.
\end{main-lemma}
\begin{proof}
 Let $m$ be the maximal index such that $\sum_{i=1}^m p_i \leq 2^{-a}$. If $m = n$ then we are done, so suppose that $m < n$. Let $S = \sum_{i=1}^m p_i$. We would like to show that $S = 2^{-a}$.
 
 The condition $p_1 \leq \cdots \leq p_n$ implies that $a_{m+1} \geq \cdots \geq a_1$, and so $k \defeq 2^{a_{m+1}} S = \sum_{i=1}^m 2^{a_{m+1}-a_i}$ is an integer. By assumption $k \leq 2^{a_{m+1}-a}$ whereas $k + 1 = 2^{a_{m+1}} \sum_{i=1}^{m+1} p_i > 2^{a_{m+1}-a}$. Since $2^{a_{m+1}-a}$ is itself an integer (since $a_{m+1} \geq a_1 \geq a$), we conclude that $k = 2^{a_{m+1}-a}$, and so $S = 2^{-a}$.
 
 \smallskip

 To prove the furthermore part, notice that by repeated applications of the  of the lemma we can partition $[n]$ into intervals whose probabilities are $2^{-a}$. The last such interval provides the required index~$\ell$. 
\end{proof}

\section{Comparisons and equality tests} \label{sec:comparison-equality}

Let $\pi$ be a distribution over $X_n = \{ x_1,\ldots,x_n \}$.
A fundamental result in information theory is that the entropy of 
a distribution $\pi $ captures the average number of queries needed to identify 
a random $x\sim\pi$.
More specifically, every algorithm asks at least $H(\pi)$ questions in expectation, and there are
algorithms that ask at most $H(\pi) + 1$ questions on average (such as Huffman coding and Shannon--Fano coding).
However, these algorithms
may potentially use arbitrary questions.

In this section we are interested in the setting where $X_n$ is linearly ordered: $x_1 \prec x_2 \prec \dots \prec x_n$.
We wish to use questions that are compatible with the ordering.
Perhaps the most natural question in this setting is a comparison query;
namely a question of the form ``$x \prec x_i$?''.
Gilbert and Moore~\cite{GilbertMoore} showed that there exists an algorithm that
uses at most $H(\pi) +  2$ comparisons. Is this tight?
Can comparison queries achieve the benchmark of $H(\pi) +  1$?

A simple argument shows that their result is tight:
let $n = 3$, and let $\pi$ be a distribution such that
\[ \pi(x_1) = \eps/2,~~~~ \pi(x_2)=1-\eps,~~~~ \pi(x_3) = \eps/2, \]
for some small $\eps$. Note that $H(\pi) = O(\eps\log\bigl(1/\eps\bigr))$, and therefore
any algorithm with redundancy 1 must use the query ``$x=x_2$?'' as its first question.
This is impossible if we only allow comparison queries (see Lemma~\ref{l52} for a more detailed and general argument).
In fact, this shows that any set of questions that achieves redundancy $1$ must include all
equality queries. So, we need to at least add all equality queries.
Is it enough? Do comparison and equality queries achieve redundancy of at most 1?

Our main result in this section gives an affirmative answer to this question:

\begin{main-theorem} \label{thm:redundancy-1}
Let $\qeq^{(n)} = \{ \{ x_i \} : 1 \leq i \leq n \}$ and $\qcomp^{(n)} = \{ \{x_1,\ldots,x_i\} : 1 \leq i \leq n-1 \}$. In other words, $\qeq^{(n)}$ consists of the questions ``$x = x_i$?'' for $i \in \{1,\ldots,n\}$, and $\qcomp^{(n)}$ consists of the questions ``$x \prec x_i$?'' for $i \in \{2,\ldots,n\}$. (Recall that $x$ is the \unknown element.)

For all $n$, $\rent(\qcomp^{(n)} \cup \qeq^{(n)}) = 1$. 
\end{main-theorem}
We prove the theorem by modifying the weight-balancing algorithm of Rissanen~\cite{Rissanen}, which uses only comparison queries and achieves redundancy~$2$ (as shown by Horibe~\cite{Horibe}).

The original weight-balancing algorithm is perhaps the first algorithm that comes to mind: it asks the most balanced comparison query (the one that splits the distribution into two parts whose probability is as equal as possible), and recurses according to the answer.

Our modified algorithm, Algorithm~$A_t$, first checks whether the probability of the most probable element $x_{\max}$ exceeds the threshold~$t$. If so, it asks the question ``$x=x_{\max}$?''. Otherwise, it proceeds as in the weight-balancing algorithm. The choice $t=0.3$ results in an algorithm whose redundancy is at most~$1$.

While a naive implementation of Algorithm~$A_t$ takes time $O(n)$ to determine the next query, this can be improved to $O(\log n)$, given $O(n)$ preprocessing time. Moreover, the entire strategy can be computed in time $O(n\log n)$. The interested reader can find the complete details in~\cite{DaganThesis}.

\subsection{Algorithm \texorpdfstring{$A_t$}{At}}
Let $\pi$ be a distribution over $X_n$. 
Let $\xmax$ denote the most probable element, and let $\pimax$ denote its probability.
Let $\xmid$ denote a point $x_i$ that minimizes $|\pi(\{x: x \prec x_i\})-1/2|$ over $i\in [n]$. We call $\xmid$ the middle\footnote{Note that one of $\{\xmid,x_{\Mid-1}\}$ is a median.} of $\pi$. 
Note that the query ``$x\prec \xmid$?'' is the most balanced query among all queries of the form ``$x \prec x_i$?''.

Let $A\subseteq\{x_1,\ldots,x_n\}$.  We use $\pi_A$ to denote $\pi(A)$; i.e.\ the probability of $A$.
Specifically, we use $\pi_{\prec x_i}, \pi_{\succeq x_i}, \pi_{\neq x_i}$ to denote $\pi_A$ when $A$ is $\{x: x \prec x_i\}, \{x: x \succeq x_i\}, \{x : x\neq x_i\}$.
We use $\pi|_A$ to denote the restriction of $\pi$ to $A$; i.e., the distribution derived by conditioning $\pi$ on $A$. 
Specifically, we use $\pi|_{\prec x_i}, \pi|_{\succeq x_i}, \pi|_{\neq x_i}$ to denote $\pi|_A$ when $A$ is $\{x: x \prec x_i\}, \{x: x \succeq x_i\}, \{x : x\neq x_i\}$. 

\begin{algorithm-description}{$A_t$}
Given a threshold $t\in (0,1)$, Algorithm~$A_t$ takes as input a distribution $\pi$ over $X_n$ and \anunknown element $x$, and  determines $x$ using only comparison and equality queries, in the following recursive fashion:
\begin{enumerate}
 \item If $\pi(x_i) = 1$ for some element $x_i$, then output $x_i$.
 \item If $\pimax \geq{t}$ then ask whether $x = \xmax$, and either output $\xmax$, or continue with $\pi|_{\neq \xmax}$.
 \item If $\pimax < {t}$, ask whether $x \prec \xmid$, and continue with either $\pi|_{\prec \xmid}$ or $\pi|_{\succeq\xmid}$.
\end{enumerate}
(When recursing on a domain $D$, we identify $D$ with $X_{|D|}$.)
\end{algorithm-description}

The weight balancing algorithm of Rissanen~\cite{Rissanen} is the special case $t = 1$; in this case no equality queries are needed, and the resulting redundancy is~$2$, as shown by Horibe~\cite{Horibe}.

We will show that for a certain range of values of ${t}$ (for example, ${t} = 0.3$), Algorithm~$A_t$ achieves redundancy at most~$1$, thus proving Theorem~\ref{thm:redundancy-1}. 

\smallskip

Recall that $A_{{t}}(\pi)$ is the cost of $A_t$ on $\pi$, and let $R_t(\pi) \defeq A_t(\pi) - H(\pi)-1$.
It is more convenient to present the proof in terms of $R_t(\pi)$ rather than in terms of the redundancy.
Our goal, stated in these terms,  is showing that 
there exists some $t$ for which $R_t(\pi)\leq 0$ for all distributions $\pi$.

We next observe two simple properties of the algorithm $A_t$. The proof of Theorem~\ref{thm:redundancy-1} relies only on these properties.

The first property is a recursive definition of $R_t$ that is convenient to induct on. See Figure~\ref{fig:recdef} for a pictorial illustration.

\providecommand{\thet}{t} 
\providecommand{\theR}{R_t} 
\begin{figure}
\begin{subfigure}{.5\linewidth}
\begin{tikzpicture}[sibling distance=5cm]
\node[circle,draw,fill=black,label={[label distance=0.5cm]below:\small $1-h(\pi_{\succeq \xmid})$},
label={[label distance=0.15cm]above:\small ``$x\prec\xmid$?''}] {}
    child{ 
    node[circle, draw, fill=black] {}
    node[itria] { $\theR(\pi|_{ \prec \xmid})$}
    edge from parent node[near end, above=10pt] {\small $\pi_{\prec \xmid}$} 
        }
    child{ 
    node[circle, draw, fill=black] {}
    node[itria] { $\theR(\pi|_{\succeq \xmid})$}
    edge from parent node[near end, above=10pt] {\small $\pi_{\succeq \xmid}$} 
        };
\end{tikzpicture}
\caption{$\pimax\in (0,\thet)$} \label{fig-lem-b-t}
\end{subfigure}
\begin{subfigure}{.5\linewidth}
\begin{tikzpicture}[sibling distance=5cm]
\node[circle,draw,fill=black, label={[label distance=0.5cm]below:\small $1-h(\pi_{\neq \xmax})$},label={[label distance=0.15cm]above:\small ``$x\neq\xmax$?''}] {}
    child{ 
    node[circle, draw, label={[label distance=0.125cm]below:\small $\theR(\pi|_{=\xmax})=-1$}] {\small $\xmax$}
    edge from parent node[near end, above=10pt] {\small $\pimax$}
        }
    child{ 
    node[circle, draw, fill=black] {}
    node[itria] {\small  $\theR(\pi|_{ \neq \xmax})$}
    edge from parent node[near end, above=10pt] {\small $1-\pimax$}
        }
;
\end{tikzpicture}
\caption{$\pimax\in [\thet,1)$} \label{fig-lem-b-m}
\end{subfigure}
\caption{Recursive definition of $\theR$}\label{fig:recdef}
\end{figure}
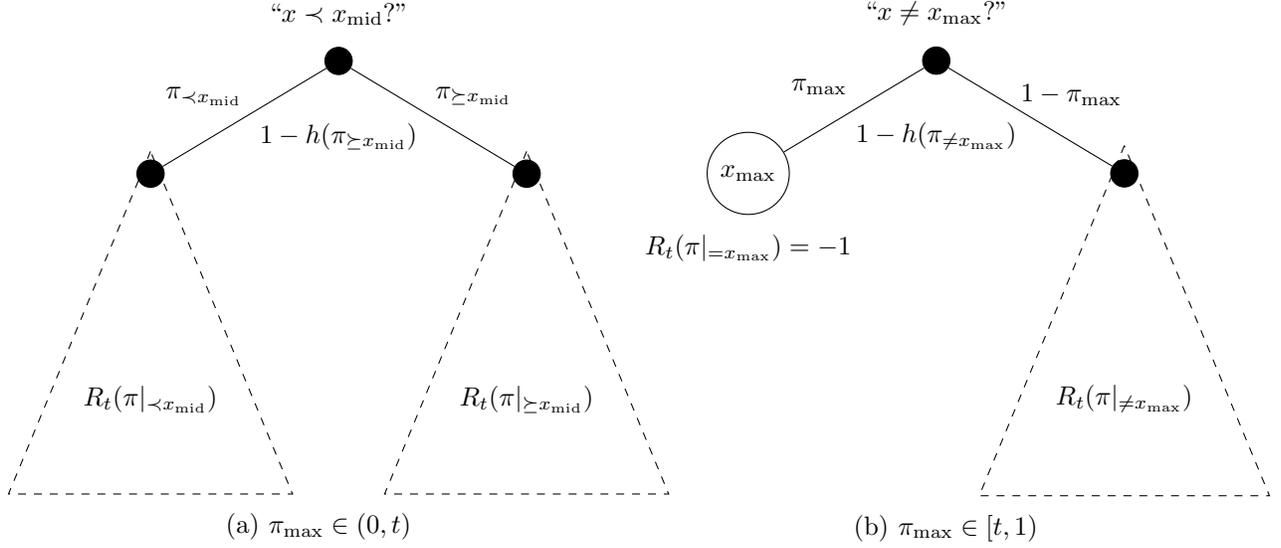

\begin{lemma} \label{lem:redundancy-1-formula}
Let $\pi$ be a distribution over $X_n$. Then
\[
 R_t(\pi) =
\begin{cases}
 -1  & \text{if } \pimax = 1, \\
 1 - h(\pimax) - \pimax + (1-\pimax) R_t(\pi|_{\neq \xmax}) & \text{if } \pimax\in[t,1), \\
 1 - h(\pi_{\prec\xmid}) + \pi_{\prec\xmid} R_t(\pi|_{\prec\xmid}) + \pi_{\succeq\xmid} R_t(\pi|_{\succeq\xmid}) & \text{if } \pimax\in(0,t).
\end{cases}
\]
\end{lemma}
\begin{proof}
If $\pimax = 1$ then $A_t(\pi) = 0$ and $H(\pi) = 0$, so $R_t(\pi) = -1$.

If $\pimax \in [t,1)$ then
\begin{align*}
 R_t(\pi) &= A_t(\pi) - H(\pi) - 1 \\ &=
 \Bigl[1 + (1-\pimax)A_t(\pi|_{\neq \xmax})\Bigr] - \Bigl[h(\pimax) + (1-\pimax)H(\pi|_{\neq \xmax})\Bigr] - \Bigl[\pimax + (1-\pimax)\Bigr] \\ &=
 1 - h(\pimax) - \pimax + (1-\pimax) \Bigl[A_t(\pi|_{\neq \xmax}) - H(\pi|_{\neq \xmax}) - 1\Bigr] \\ &=
 1 - h(\pimax) - \pimax + (1-\pimax) R_t(\pi_{\neq \xmax}).
\end{align*}

If $\pimax \in (0,t)$ then
\begin{align*}
 R_t(\pi) &= A_t(\pi) - H(\pi) - 1 \\ &=
 \Bigl[1 + \pi_{\prec\xmid} A_t(\pi|_{\prec\xmid}) + \pi_{\succeq\xmid} A_t(\pi|_{\succeq\xmid})\Bigr] \\&\qquad -
 \Bigl[h(\pi_{\prec \xmid}) + \pi_{\prec\xmid} H(\pi|_{\prec\xmid}) + \pi_{\succeq\xmid} H(\pi|_{\succeq\xmid})\Bigr] -
 \Bigl[\pi_{\prec\xmid} + \pi_{\succeq\xmid}\Bigr] \\ &=
 1 - h(\pi_{\prec \xmid}) + \pi_{\prec\xmid} \Bigl[A_t(\pi|_{\prec\xmid}) - H(\pi|_{\prec\xmid}) - 1\Bigr] + \pi_{\succeq\xmid} \Bigl[A_t(\pi|_{\succeq\xmid}) - H(\pi|_{\succeq\xmid}) - 1\Bigr] \\ &=
 1 - h(\pi_{\prec\xmid}) + \pi_{\prec\xmid} R_t(\pi|_{\prec\xmid}) + \pi_{\succeq\xmid} R_t(\pi|_{\succeq\xmid}). \qedhere
\end{align*}
\end{proof}

The second property is that whenever $A_t$ uses a comparison query (i.e.\ when $\pimax < t$), then this question is balanced:
\begin{lemma}\label{l412}
Let $\pi$ be a distribution over $X_n$. Then $\pi_{\prec\xmid},\pi_{\succeq\xmid}\in[\frac{1-\pimax}{2},\frac{1+\pimax}{2}]$.
\end{lemma}
\begin{proof}
By the definition of $\xmid$, it suffices to show that there exists some $j$ with $\pi_{\prec j}\in [\frac{1-\pimax}{2},\frac{1+\pimax}{2}]$.
Indeed, for all $j$, $\pi_{\prec j+1} - \pi_{\prec j} = \pi_j\leq\pimax$, and therefore, if $j$ is the maximum element with $\pi_{\prec j}<\frac{1}{2}$ (possibly $j=0$), then either $\pi_{\prec j}$ or $\pi_{\prec j+1}$ are in  $[\frac{1-\pimax}{2},\frac{1+\pimax}{2}]$. 
\end{proof}

\subsection{Game \texorpdfstring{$G_t$}{Gt}}

We will use the properties described in Lemma~\ref{l412} to bound $R_t(\pi)$.
Since this involves quantifying over all distributions $\pi$, it is convenient to introduce a game that simulates $A_t$ on a distribution $\pi$.
The game involves one player, Alice, who we think of as an adversary that chooses the input distribution $\pi$ (in fact, she only chooses $\pimax$), and wins a revenue of $R_t(\pi)$ (thus, her objective is to maximize the redundancy). This reduces our goal to showing that Alice's optimum revenue is nonpositive. 
The definition of the game is tailored to the properties stated in Lemma~\ref{lem:redundancy-1-formula} and Lemma~\ref{l412}. 
We first introduce $G_t$, and then relate it to the redundancy of $A_t$ (see Lemma~\ref{l416} below).

\begin{algorithm-description}[Game]{$G_t$}
Let $t\leq\frac{1}{3}$, and let $f,s\colon (0,1]\rightarrow \mathbb{R}$ be
$f(x) \defeq 1-h(x) - x$ and $s(x) \defeq 1-h(x)$.
The game $G_t$ consists of one player called Alice, 
whose objective is to maximize her revenue. 
The game $G_t$ begins at an initial state $p\in (0,1]$, and proceeds as follows.
\begin{enumerate}
\item If $p\in [t,1]$, the game ends with revenue $f(p)$.
\item If $p\in (0,t)$, then Alice chooses a state\footnote{Note that $p'\in(0,1]$, since $p<t\leq\frac{1}{3}$ implies that $\bigl[\frac{2p}{1+p}, \frac{2p}{1-p}\bigr]\subseteq(0,1]$.} $p'\in\bigl[\frac{2p}{1+p}, \frac{2p}{1-p}\bigr]$ and
recursively plays $G_t$ with initial state $p'$. Let $r'$ denote her revenue in the game that begins at $p'$.
 Alice's final revenue is
\[ s\biggl(\frac{p}{p'}\biggr) + \frac{p}{p'}\cdot r'. \eofahere \] 
\end{enumerate}
\end{algorithm-description}

Note that given any initial state $p$ and a strategy for Alice, the game $G_t$ always terminates:
indeed, if $p=p_0,p_1,p_2,\ldots$ is the sequence of states chosen by Alice,
then as long as $p_i < t$, it holds that $p_{i+1} \geq \frac{2p_i}{1+p_i} > \frac{2p_i}{1+1/3}=\frac{3}{2}p_i$ (the second inequality is since $p_i < t\leq\frac{1}{3}$). 
So the sequence of states grows at an exponential rate, which means that for $\ell=O(\log(1/p_0))$, the state $p_\ell$ exceeds the threshold $t$ and the game terminates. 

For $p\in (0,1]$, let $r_t(p)$ denote the supremum of Alice's revenue in $G_t$ when the initial state is $p$, the supremum ranging over all possible strategies for Alice. 

Our next step is using the game $G_t$ to prove Theorem~\ref{thm:redundancy-1}:
We will show that $t=0.3$ satisfies:
\begin{enumerate}[label=(\roman*)]
\item $r_t(p)\leq 0$ for all $p\in(0,1]$, and
\item $R_t(\pi)\leq r_t(\pimax)$ for all $\pi$.
\end{enumerate}
Note that (i) and (ii) imply Theorem~\ref{thm:redundancy-1}.
Before establishing (i) and (ii), we state and prove three simple lemmas regarding $G_t$ 
that are useful to this end.

The first lemma will be used in the proof of Lemma~\ref{l416}, which shows
that if $r_t(p)\leq 0$ for all $p\in (0,1]$, then $R_t(\pi)\leq r_t(\pimax)$.
Its proof follows directly from the definition of $r_t$.
\begin{lemma}\label{l413}
For all $p < t$ and all $p'\in\big[\frac{2p}{1+p}, \frac{2p}{1-p}\big]$:
\[r_t(p) \geq s\left(\frac{p}{p'}\right) + \frac{p}{p'}\cdot r_t(p').\]
\end{lemma}

The next two lemmas will be used in the proof of Lemma~\ref{l417}, which shows
that $r_t(p)\leq 0$ for all $p\in (0,1]$ when $t=0.3$.
The first one gives a tighter estimate on the growth of the sequence of states: 
\begin{lemma}\label{l414}
Let $p_0,p_1,\ldots,p_k$ be the sequence of states chosen by Alice. For every $i\leq k-1$:
\[ p_{k-1-i} < \frac{t}{2^i(1-t)+t}. \]
\end{lemma}
\begin{proof}
We prove the bound by induction on $i$. 
The case $i=0$ follows from $p_{k-1}$ being a non-final state, and therefore $p_{k-1} < t$ as required.
Assume now that $i > 0$.
By the definition of $G_t$ it follows that $p_{k-1-(i-1)}\in\bigl[\frac{2p_{k-1-i}}{1+p_{k-1-i}}, \frac{2p_{k-1-i}}{1-p_{k-1-i}}\bigr]$, which implies that $p_{k-1-i} \in \bigl[\frac{p_{k-1-(i-1)}}{2+p_{k-1-(i-1)}}, \frac{p_{k-1-(i-1)}}{2-p_{k-1-(i-1)}}\bigr]$.
Therefore,
\begin{align*}
p_{k-1-i} &\leq \frac{p_{k-1-(i-1)}}{2-p_{k-1-(i-1)}}\\
	      &< \frac{t/\bigl(2^{i-1}(1-t)+t\bigr)}{2-t/\bigl(2^{i-1}(1-t)+t\bigr)}\tag{by induction hypothesis on $i-1$}\\
	      &=\frac{t}{2^i(1-t)+t}. \qedhere
\end{align*}
\end{proof}

The second lemma gives a somewhat more explicit form of the revenue of $G_t$:
\begin{lemma}\label{l415}
Let $p_0,p_1,\ldots, p_k$ be the sequence of states chosen by Alice. Let $r(p_0,\ldots,p_k)$ denote the revenue obtained by choosing these states. Then
\[r(p_0,\ldots,p_k) = \sum_{i=0}^{k-1}\frac{p_0}{p_i} s\biggl(\frac{p_{i}}{p_{i+1}}\biggr) + \frac{p_{0}}{p_k}f(p_k).\]
\end{lemma}
\begin{proof}
We prove the formula by induction on $k$. If $k=0$ then $p_k \geq t$ and $r(p_k) = f(p_k) = \frac{p_0}{p_k}f(p_k)$.
When $k\geq 1$:
\begin{align*}
r(p_0,\ldots,p_k) &= s\biggl(\frac{p_0}{p_1}\biggr) + \frac{p_0}{p_1}\cdot r(p_1,\dots, p_k)\tag{by definition of $G_t$}\\
			  &= \frac{p_0}{p_0} s\biggl(\frac{p_0}{p_1}\biggr) + \frac{p_0}{p_1}\cdot \Bigl(\sum_{i=1}^{k-1}\frac{p_1}{p_i} s\biggl(\frac{p_{i}}{p_{i+1}}\biggr) + \frac{p_{1}}{p_k}f(p_k)\Bigr)\tag{by induction hypothesis}\\
			  &= \sum_{i=0}^{k-1}\frac{p_0}{p_i} s\biggl(\frac{p_{i}}{p_{i+1}}\biggr) + \frac{p_{0}}{p_k}f(p_k). \qedhere
\end{align*}
\end{proof}

\subsection{Relating \texorpdfstring{$A_t$}{At} to \texorpdfstring{$G_t$}{Gt}}
Next, we relate the revenue in $G_t$ to the redundancy of $A_t$ by linking $r_t$ and $R_t$. 
We first reduce the Theorem~\ref{thm:redundancy-1} to showing that there exists some $t\in (0,1]$ such that
$r_t(p)\leq 0$ for all $p\in (0,1]$ (Lemma~\ref{l416}), and then show that $t = 0.3$ satisfies this condition (Lemma~\ref{l417}).
\begin{lemma}\label{l416}
Let $t$ be such that $r_t(p)\leq 0$ for all $p\in(0,1]$. For every distribution $\pi$,
$$R_t(\pi)\leq r_t(\pimax).$$
In particular, such $t$ satisfies $R_t(\pi)\leq 0$ for all $\pi$.
\end{lemma}
\begin{proof}
We proceed by induction on the size of $\supp{\pi}=\{x_i : \pi(x_i)\neq 0\}$.
If $\bigl\lvert\supp{\pi}\bigr\rvert = 1$ then $\pimax=1$, and therefore $R_t(\pi)=-1$, $r_t(\pimax)=0$, and indeed $R_t(\pi)\leq r_t(\pimax)$.
Assume now that $\bigl\lvert\supp{\pi}\bigr\rvert = k > 1$. Since $k>1$, it follows that $\pimax < 1$. There are two cases, according to whether $\pimax\in(0,t)$ or $\pimax\in [t,1)$.

If $\pimax\in(0,t)$ then $A_t$ asks whether $x \prec \xmid$, and continues accordingly with $\pi|_{\prec\xmid}$
or $\pi|_{\succeq\xmid}$. Let $\sigma \defeq \pi|_{\prec\xmid}$ and $\tau \defeq \pi|_{\succeq\xmid}$.
By Lemma~\ref{lem:redundancy-1-formula}:
\begin{align*}
R_t(\pi) &=  1 - h(\pi_{\prec\xmid}) + \pi_{\prec\xmid} R_t(\sigma) + \pi_{\succeq\xmid} R_t(\tau)\tag{since $\pimax\in(0,t)$}\\
	   &\leq 1 - h(\pi_{\prec\xmid}) + \pi_{\prec\xmid} r_t(\sigma_{\max}) + \pi_{\succeq\xmid} r_t(\tau_{\max})\tag{by induction hypothesis}
\end{align*}
Without loss of generality, assume that $\xmax \prec \xmid$. Therefore $\sigma_{\max} = {\pimax}/{\pi_{\prec\xmid}}$, and 
by Lemma ~\ref{l412}:
\begin{equation}\label{eq1}
\sigma_{\max}\in \left[\frac{2\pimax}{1+\pimax},\frac{2\pimax}{1-\pimax}\right]. 
\end{equation}
Thus,
\begin{align*}
R_t(\pi) &\leq 1 - h(\pi_{\prec\xmid}) + \pi_{\prec\xmid} r_t(\sigma_{\max})\tag{since $r_t(\tau_{\max})\leq 0$}\\
	 &= 1 - h\left(\frac{\pimax}{\sigma_{\max}}\right) + \frac{\pimax}{\sigma_{\max}}r_t(\sigma_{\max})\tag{$\sigma_{\max} = \frac{\pimax}{\pi_{\prec\xmid}}$}\\
	 &=s\left(\frac{\pimax}{\sigma_{\max}}\right) + \frac{\pimax}{\sigma_{\max}}r_t(\sigma_{\max})\tag{by definition of $s$}\\
	 &\leq r_t(\pimax).\tag{by~\eqref{eq1} and Lemma~\ref{l413}}
\end{align*}

The analysis when $\pimax \in [t,1)$ is very similar. In this case $A_t$ asks whether $x = \xmax$, and continues with $\pi|_{\neq\xmax}$ if $x \neq \xmax$. Let $\sigma \defeq \pi|_{\leq \xmax}$. By Lemma~\ref{lem:redundancy-1-formula},
\begin{align*}
 R_t(\pi) &= 1 - h(\pimax) - \pimax + (1-\pimax) R_t(\sigma)	 \tag{since $\pimax \in (t,1)$} \\ &\leq
 1 - h(\pimax) - \pimax + (1-\pimax) r_t(\sigma_{\max}) \tag{by induction hypothesis} \\ &\leq
 1 - h(\pimax) - \pimax \tag{since $r_t(\sigma_{\max}) \leq 0$} \\ &=
 f(\pimax) \tag{by definition of $f$} \\ &=
 r_t(\pimax). \tag{by definition of $r_t$, since $\pimax \geq t$}
\end{align*}
\end{proof}

The following lemma shows that $t=0.3$ satisfies $r_t(p)\leq 0$ for all $p\in (0,1]$,
completing the proof of Theorem~\ref{thm:redundancy-1}. It uses some technical results, proved below in Lemma~\ref{la01}.
\begin{lemma}\label{l417}
Let $t=0.3$. Then $r_{t}(p)\leq 0$ for all $p\in(0,1]$.
\end{lemma}

\begin{proof}
Let $p\in (0,1]$.
We consider two cases: (i) $p\geq t$, and (ii) $p < t$.
In each case we derive a constraint on $t$ that
suffices for ensuring that $r_t(p) \leq 0$, and conclude the proof 
by showing that $t= 0.3$ satisfies both constraints.

Consider the case $t\leq p$. Here $r_t(p) = f(p) = 1-h(p)-p$,
and calculation shows that $f(p)$ is non-positive on $[0.23,1]$;
therefore, $r_t(p)\leq 0$ for all $p\geq t$, as long as $t\geq 0.23$.

Consider the case $p < t$. Here, we are not aware of an explicit
formula for $r_t(p)$; instead, we derive the following upper bound, for all $p < t$:
\begin{equation}\label{eq2}
\frac{r_t(p)}{p} \leq \sum_{n=0}^\infty S\left(\frac{t}{2^n(1-t)+t}\right) + \max\left(F(t),F\Bigl(\frac{2t}{1-t}\Bigr)\right),
\end{equation}
where \[ S(x)=\frac{s\bigg(\displaystyle\frac{1+x}{2}\bigg)}{x}, ~~~~~ F(x)=\frac{f(x)}{x}. \]
With~\eqref{eq2} in hand we are done: indeed, calculation shows that the right hand side of~\eqref{eq2} is non-positive in some neighborhood of $0.3$ (e.g.\ it is $-0.0312$ when $t=0.3$, it is $-0.0899$ when $t = 0.294$); thus, as these values of $t$ also satisfy the constraint from (i), this finishes the proof.

It remains to prove~\eqref{eq2}.
Let $p=p_0,p_1,p_2,\ldots,p_k$ be a sequence of states chosen by Alice.  
It suffices to show that $\frac{r(p_0,\ldots,p_k)}{p_0}\leq\sum_{n=0}^\infty S\left(\frac{t}{2^n(1-t)+t}\right) + \max\left(F(t),F\Bigl(\frac{2t}{1-t}\Bigr)\right).$
By Lemma~\ref{l415}:
\begin{align*}
r(p_0,\ldots,p_k) &= \sum_{i=0}^{k-1}\frac{p_0}{p_i} s\biggl(\frac{p_{i}}{p_{i+1}}\biggr) + \frac{p_{0}}{p_k}f(p_k)\\
			&\leq \sum_{i=0}^{k-1}\frac{p_0}{p_i} s\biggl(\frac{1+p_{i}}{2}\biggr) + \frac{p_{0}}{p_k}f(p_k)\\
			&=p_0\Bigl(\sum_{i=0}^{k-1}{S(p_i)} + F(p_k)\Bigr),
\end{align*}
where in the second line we used that $\frac{p_{i}}{p_{i+1}}\in[\frac{1-p_i}{2}, \frac{1+p_i}{2}]$, and the fact that $s(x)=1-h(x)$ is symmetric around $x=0.5$ and increases with $\lvert x-0.5 \rvert$. Therefore,
\begin{align*}
\frac{r(p_0,\ldots,p_k)}{p_0} &\leq \sum_{i=0}^{k-1}{S(p_i)} + F(p_k)\\
					  &\leq \sum_{i=0}^{k-1} S\left(\frac{t}{2^i(1-t)+t}\right) + F(p_k)\\
          &\leq \sum_{i=0}^\infty S\left(\frac{t}{2^i(1-t)+t}\right) + \max\left(F(t),F\Bigl(\frac{2t}{1-t}\Bigr)\right),
\end{align*}
where in the second last inequality we used that $p_{k-1-i} < \frac{t}{2^i(1-t)+t}$ (Lemma~\ref{l414}) and that $S(x)$ is monotone (Lemma~\ref{la01} below), and in the last inequality we used that $p_k\in [t, \frac{2t}{1-t})$ and that $F(x)$ is convex (Lemma~\ref{la01} below).
\end{proof}

The following technical lemma completes the proof of Lemma~\ref{l417}.

\begin{lemma} \label{la01}
 The function $S(x) = \frac{1-h(\frac{1+x}{2})}{x}$ is monotone, and the function $F(x) = \frac{1-h(x)-x}{x}$ is convex.
\end{lemma}
\begin{proof}
The function $h(\frac{1-x}{2})$ is equal to its Maclaurin series for $x\in (-1,+1)$:
\[ h\Bigl(\frac{1+x}{2}\Bigr) = 1 - \sum_{k=1}^{\infty}{\frac{\log_2{e}}{2k(2k-1)}\cdot x^{2k}}. \]
Therefore,
\[ S(x) = \sum_{k=1}^{\infty}{\frac{\log_2{e}}{2k(2k-1)}\cdot x^{2k-1}}, \]
and 
\[ F(x) =  \Bigl(\sum_{k=1}^{\infty}{\frac{\log_2{e}}{2k(2k-1)}}\cdot\frac{(1-2x)^{2k}}{x}\Bigr) - 1. \]
Now, each of the functions $x^{2k-1}$ is monotone, and each of the functions 
$\frac{(1-2x)^{2k}}{x}$ is convex on $(0,\infty)$: its second derivative is
\[
 \frac{2(1-2x)^{2k-2}(1 + 4(k-1)x + 4(k-1)(2k-1)x^2)}{x^3} > 0.
\]
Therefore, $S(x)$ is monotone as a non-negative combination of monotone functions,
and $F(x)$ is convex as a non-negative combination of convex functions.
\end{proof}

\section{Information theoretical benchmark --- Shannon's entropy} \label{sec:redundancy-r}
In this section we study the minimum number of questions that achieve redundancy of at most $r$, for a fixed $r \geq 1$. Note that $r=1$ is the optimal redundancy: the distribution $\pi$ on $X_2$ given by $\pi_1 = 1-\epsilon,\allowbreak \pi_2=\epsilon$ has redundancy $1-\tilde{O}(\epsilon)$ (that is, $1-O(\epsilon \log(1/\epsilon))$ even without restricting the set of allowed questions.

In the previous section we have shown that the optimal redundancy of $r=1$ can be achieved with just $2n$ comparison and equality queries (in fact, as we show below, there are only $2n-3$ of these queries).
It is natural to ask how small the number of questions can be if we allow for a larger $r$. 
Note that at least $\log n$ questions are necessary to achieve any finite redundancy. Indeed, a smaller set of questions is not capable of specifying all elements even if all questions are being asked.

The main result of this section is that the minimum number of questions that are sufficient for achieving
redundancy $r$ is roughly $r\cdot n^{1/ \lfloor r \rfloor}$:

\begin{main-theorem} \label{thm:redundancy-r-bounds}
For every $r \geq 1$ and $n\in\mathbb{N}$,
\[
 \frac{1}{e}\lfloor r \rfloor n^{1/\lfloor r \rfloor} \leq \uent(n,r) \leq 2 \lfloor r \rfloor n^{1/\lfloor r \rfloor}.
\]
In particular, $\uent(n,r) = \Theta\bigl(\lfloor r\rfloor n^{1/\lfloor r\rfloor}\bigr)$.
\end{main-theorem}


\subsection{Upper bound} \label{sec:redundancy-r-ub}
The upper bound in Theorem~\ref{thm:redundancy-r-bounds} is based on the following corollary of Theorem~\ref{thm:redundancy-1}:

\begin{theorem}\label{thm:redundancy-r}
Let $Y$ be a linearly ordered set, and let $Z=Y^k$ (we don't think of $Z$ as ordered).

For any distribution $\pi$ on $Z$ there is an algorithm
that uses only questions of the form (i) ``$\vec x_i \prec y$?'' and (ii) ``$\vec x_i = y$?'', where $i\in[k]$ and $y\in Y$,
whose cost is at most $H(\pi) + k$.
\end{theorem}

\begin{proof}
Let $Z_1Z_2\ldots Z_k \sim \pi$. Consider the algorithm which determines $Z_1,\ldots,Z_k$ in order, where $Z_i$ is determined by applying the algorithm from Theorem~\ref{thm:redundancy-1} on ``$Z_i \vert Z_1\ldots Z_{i-1}$", which is the conditional distribution of $Z_i$ given the known values of $Z_1,\ldots,Z_{i-1}$.
The expected number of queries is at most
\[
 \bigl(H(Z_1) + 1\bigr) + \bigl(H(Z_2 \vert Z_1) + 1\bigr) + \cdots + \bigl(H(Z_k\vert Z_1\ldots Z_{k-1}) + 1\bigr) =
 H(Z_1\ldots Z_k) + k,
\]
using the chain rule.
\end{proof}

We use this theorem to construct a set of questions of size at most $2 \lfloor r \rfloor n^{1/\lfloor r \rfloor}$
that achieves redundancy $r$ for any distribution over $X_n$.

Note that $n \leq \Bigl(\bigl\lceil n^{1/\lfloor r\rfloor}\bigr\rceil\Bigr)^{\lfloor r\rfloor}$.
Therefore every element $x\in X_n$ can be represented by a vector $\vec x \in \allowbreak \bigl\{1,\dots,\lceil n^{1/\lfloor r\rfloor}\rceil \bigr\} ^ {\lfloor r\rfloor}$.
Let $\cQ$ be the set containing all questions of the form (i) ``$\vec x_i = y$?'' 
and (ii) ``$\vec x_i \prec y$?''. 
By Theorem~\ref{thm:redundancy-r}, $r(\cQ)=\lfloor r\rfloor$.

The following questions from $\cQ$ are redundant, and can be removed from $\cQ$ without increasing its redundancy:
(i) ``$\vec x_i \prec 1$?'' (corresponds to the empty set and therefore provides no information), (ii) ``$\vec x_i \prec 2$?'' (equivalent to the question ``$\vec x_i = 1$?''), and (iii) ``$\vec x_i \prec \lceil n^{1/\lfloor r\rfloor}\rceil$?'' (equivalent to the question ``$\vec x_i = \lceil n^{1/\lfloor r\rfloor}\rceil$?'').
The number of remaining questions is
\[ \lfloor r \rfloor\cdot\Bigl(2\bigl\lceil n^{1/\lfloor r\rfloor}\bigr\rceil - 3\Bigr) \leq 2\lfloor r \rfloor\cdot\Bigl( n^{1/\lfloor r\rfloor} \Bigr). \]
This proves the upper bound in Theorem~\ref{thm:redundancy-r-bounds}.

\subsection{Lower bound} \label{sec:redundancy-r-lb}
The crux of the proof of the lower bound in Theorem~\ref{thm:redundancy-r-bounds} is that if $\cQ$
is a set of questions whose redundancy is at most $r$ then every $x\in X_n$
can be identified by at most $\lfloor r\rfloor$ questions from $\cQ$.

We say that the questions $q_1,\ldots,q_T$ \emph{identify} $x$
if for every $y\neq x$ there is some $i\leq T$ such that
$q_i(x)\neq q_i(y)$. Define $t(n,r)$ to be the minimum cardinality 
of a set $\cQ$ of questions such that every $x\in X$ has at most
$r$ questions in $\cQ$ that identify it.
The quantity $t(n,r)$ can be thought of as a non-deterministic version of $u(n,r)$:
it is the minimal size of a set of questions so that every element can be ``verified''
using at most $r$ questions.

The lower bound on $u(n,r)$ follows from Lemma~\ref{l52} and Lemma~\ref{l53} below. 
\begin{lemma}\label{l52}
For all $n,r$, $u(n,r)\geq t(n,\lfloor r \rfloor)$.
\end{lemma}
\begin{proof}
It suffices to show that for every set of questions $\cQ$ with redundancy at most $r$,
every $x\in X$ has at most $\lfloor r\rfloor$ questions in $\cQ$ that identify it.

Consider the distribution $\pi$ given by $\pi(x) = 1-\epsilon$ and $\pi(y) = \epsilon/(n-1)$ for $y \neq x$. Thus $H(\pi) = \tilde{O}(\epsilon)$.
Consider an algorithm for $\pi$ with redundancy $r$ that uses only questions from $\cQ$. Let $T$ be the number of questions it uses to find $x$. The cost of the algorithm is at least $(1-\epsilon)T$, and so $(1-\epsilon)T \leq H(\pi) + r = \tilde{O}(\epsilon) + r$, implying $T \leq \tilde{O}(\epsilon) + (1+\frac{\epsilon}{1-\epsilon})r$. For small enough $\epsilon > 0$, the right-hand side is smaller than $\lfloor r \rfloor + 1$, and so $T \leq \lfloor r \rfloor$.
\end{proof}

Lemma~\ref{l52} says that in order to lower bound $u(n,r)$, it suffices to lower bound $t(n,\lfloor r\rfloor)$, which is easier to handle. 
For example, the following straightforward argument shows that $t(n,R)\geq \frac{1}{2e}R n^{ 1/R}$, for every $R,n\in\mathbb{N}$. Assume $\cQ$ is a set of questions of size $u(n,R)$ so that every $x$ is identified by at most $R$ questions.
This implies an encoding (i.e.\ a one-to-one mapping) of $x\in X_n$ by the $R$ questions identifying it, and by the bits indicating whether $x$ satisfies each of these questions. Therefore
\begin{align*}
n \leq \binom{|\cQ|}{\leq R }2^{R}
   \leq \Bigl(\frac{2e|\cQ|}{R}\Bigr)^{R}, 
\end{align*}
where in the last inequality we used that $\binom{m}{\leq k}\leq \bigl(\frac{em}{k}\bigr)^k$ for all $m,k$.
This implies that $t(n,\lfloor r\rfloor) \geq \frac{1}{2e}\lfloor r\rfloor n^{ 1/\lfloor r \rfloor}$.
The constant $\frac{1}{2e}$ in front of $\lfloor r\rfloor n^{ 1/\lfloor r \rfloor}$ can be increased to $\frac{1}{e}$, 
using an algebraic argument:
\begin{lemma}\label{l53}
For all $n,R\in\mathbb{N}$: \[ t(n,R)\geq \frac{1}{e}R\cdot n^{1/R}. \]
\end{lemma}
\begin{proof}
We use the so-called polynomial method. 
Let $\cQ$ be a set of questions such that each $x\in X$ can be identified by at most $R$ queries.
For each $x\in X$, let $u_x$  be the $|\cQ|$-dimensional vector $u_x=\bigl(q_1(x),\ldots,q_{\lvert \cQ\rvert}(x)\bigr)$, and let $U=\{u_x : x\in X\}\subseteq\{0,1\}^{\lvert \cQ\rvert}$. We will show that every function $F\colon U \to \mathbb{F}_2$ can be represented as a multilinear polynomial of degree at most $R$ in $|\cQ|$ variables. Since the dimension over $\mathbb{F}_2$ of all such functions is $n$, whereas the dimension of the space of all multilinear polynomials of degree at most $R$ is $\binom{|\cQ|}{\leq R}$, the bound follows:
\[
 n \leq \binom{|\cQ|}{\leq R} \leq \Bigl(\frac{e|\cQ|}{R}\Bigr)^R \Longrightarrow n \geq \frac{1}{e} R \cdot n^{1/R}.
\] 
 
It is enough to show that for any $u_x\in U$, the corresponding ``delta function'' $\delta_x\colon U\to\mathbb{F}_2$, defined as $\delta_x(u_{x})=1$ and $\delta_x(v) =0$ for $u_{x}\neq v\in U$, can be represented as a polynomial of degree at most $d$.
Suppose that $q_{i_1},\ldots, q_{i_T}$ are $T\leq R$ questions that identify $x$. Consider the polynomial
\[
 P(y_1,\ldots,y_{\lvert \cQ\rvert}) = (y_{i_1} - q_{i_1}(x) + 1) \cdots (y_{i_r} - q_{i_r}(x) + 1).
\]
 Clearly $P(u_x) = 1$. On the other hand, if $P(u_y) = 1$ then $q_{i_j}(y) = q_{i_j}(x)$ for all $j$, showing that $y = x$. So $P = \delta_x$, completing the proof.
\end{proof}

Our proof of the lower bound in Theorem~\ref{thm:redundancy-r-bounds} is based
on $t(n,R)$, which is the minimum cardinality of a set of queries such that each element can
be identified by at most $R$ questions. This quantity is closely related to witness codes~\cite{Meshulam,CRZ08}; see~\cite{DaganThesis} for more details.

\section{Combinatorial benchmark --- Huffman codes} \label{sec:huffman}

Section~\ref{sec:comparison-equality} shows that the optimal redundancy, namely $1$, can be achieved using only $O(n)$ questions. However, it is natural to ask for an \emph{instance}-optimal algorithm? That is, we are looking for a set of questions which matches the performance of minimum redundancy codes such as Huffman codes.

Let us repeat the definition of an optimal set of questions that is central in this section.
\begin{main-definition} \label{dfn:optimal-set-of-questions}
A set \hitter of subsets of $X_n$ is an \emph{optimal set of questions over \groundset} if for all distributions $\mu$ on \groundset,
\[
 \costt{\hitter}{\mu} = \opt{\mu}.
\]	
\end{main-definition}

Using the above definition, $\uhuf(n,0)$ is equal to the minimal size of an optimal set of questions over $X_n$.
Perhaps surprisingly, the trivial upper bound of $2^{n-1}$ on $\uhuf(n,0)$ can be exponentially improved:
\begin{main-theorem} \label{thm:minimum-redundancy-ub}
 We have
\[
 \uhuf(n,0) \leq 1.25^{n+o(n)}.
\]
\end{main-theorem}

We prove a similar lower bound, which is almost tight for infinitely many $n$:
\begin{main-theorem} \label{thm:minimum-redundancy-lb}
 For $n$ of the form $n=5 \cdot 2^m$,
\[
 \uhuf(n,0) \geq 1.25^{n}/O(\sqrt{n}).
\]

 For all $n$,
\[
 \uhuf(n,0) \geq 1.232^n/O(\sqrt{n}).
\]
\end{main-theorem}

\begin{main-corollary} \label{cor:minimum-redundancy}
 We have
\[
 \limsup_{n\to\infty} \frac{\log \uhuf(n,0)}{n} = \log 1.25.
\]
\end{main-corollary}

Unfortunately, the construction in Theorem~\ref{thm:minimum-redundancy-ub} is not explicit. A different construction, which uses $O(\sqrt{2}^n)$ questions, is not only explicit, but can also be implemented efficiently:

\begin{main-theorem} \label{thm:cone}
 Consider the set of questions
\[
 \hitter = \{ A \subseteq X_n : A \subseteq X_{\lceil n/2 \rceil} \text{ or } A \supseteq X_{\lceil n/2 \rceil} \}.
\]
 The set $\hitter$ consists of $2^{\lceil n/2 \rceil} + 2^{\lfloor n/2 \rfloor}$ questions and satisfies the following properties:
\begin{enumerate}
 \item There is an indexing scheme $\hitter = \{ Q_q : q \in \{0,1\}^{\lceil n/2 \rceil + 1} \}$ such that given an index $q$ and an element $x_i \in X_n$, we can decide whether $x_i \in Q_q$ in time $O(n)$.
 \item Given a distribution $\pi$, we can construct an optimal decision tree for $\pi$ using \hitter in time $O(n^2)$.
 \item Given a distribution $\pi$, we can implement an optimal decision tree for $\pi$ in an online fashion in time $O(n)$ per question, after $O(n\log n)$ preprocessing.
\end{enumerate}
\end{main-theorem}

\paragraph{Section organization.} Section~\ref{sec:dyadic} shows that a set of questions is optimal if and only if it is a \emph{dyadic hitter}, that is, contains a question splitting every non-constant dyadic distribution into two equal halves. Section~\ref{sec:antichains} discusses a relation to hitting sets for maximal antichains, and proves Theorem~\ref{thm:cone}. Section~\ref{sec:un0-density} shows that the optimal size of a dyadic hitter is controlled by the minimum value of another parameter, the \emph{maximum relative density}. We upper bound the minimum value in Section~\ref{sec:un0-ub}, thus proving Theorem~\ref{thm:minimum-redundancy-lb}, and lower bound it in Section~\ref{sec:un0-lb}, thus proving Theorem~\ref{thm:minimum-redundancy-ub}.

\subsection{Reduction to dyadic hitters} \label{sec:dyadic}

The purpose of this subsection is to give a convenient combinatorial characterization of optimal sets of questions.
Before presenting this characterization, we show that in this context it suffices to look at dyadic distributions.
\begin{lemma} \label{lem:optimal-dyadic}
  A set \hitter of questions over \range is optimal if and only if $\costt{\hitter}{\mu} = \opt{\mu}$ for all \emph{dyadic} distributions $\mu$.
\end{lemma}
\begin{proof}
 Suppose that \hitter is optimal for all dyadic distributions, and let $\pi$ be an arbitrary distribution over \range. Let $\mu$ be a dyadic distribution such that
\[
 \opt{\pi} = \sum_{i=1}^n \pi_i \log \frac{1}{\mu_i}.
\]
 By assumption, \hitter is optimal for $\mu$. Let $T$ be an optimal decision tree for $\mu$ using questions from \hitter only, and let $\tau$ be the corresponding dyadic distribution, given by $\tau_i = 2^{-T(x_i)}$ (recall that $T(x_i)$ is the depth of $x_i$). Since $\tau$ minimizes $T(\mu) = H(\mu) + D(\mu\|\tau)$ over dyadic distributions, necessarily $\tau = \mu$. Thus
\[
 T(\pi) = \sum_{i=1}^n \pi_i \log \frac{1}{\tau_i} = \sum_{i=1}^n \pi_i \log \frac{1}{\mu_i} = \opt{\pi},
\]
 showing that \hitter is optimal for \dist.
\end{proof}

Given a dyadic distribution $\mu$ on \range, we will be particularly interested in the collection of subsets of \range that have probability exactly half under $\mu$.
\begin{dfn}[Dyadic hitters]\label{dfn:dyadSetsandHitters}
 Let \dist be a non-constant dyadic distribution. A set $A \subseteq X_n$ \emph{splits} \dist if $\dist(A) = 1/2$. We denote the collection of all sets splitting \dist by \dyadof{\dist}. We call a set of the form \dyadof{\dist} a \emph{dyadic set}.
 
 We call a set of questions \hitter a \emph{dyadic hitter} in $X_n$ if it intersects \dyadof{\dist} for all non-constant dyadic distributions \dist. (Lemma~\ref{lem:neat-sum} implies that \dyadof{\dist} is always non-empty.)
  \end{dfn}

  A dyadic hitter is precisely the object we are interested in:
  \begin{lemma}\label{lem:hitterIsOptimal}
   A set \hitter of subsets of \range is an optimal
   set of questions if and only if it is a \dyadhitter in \range.
  \end{lemma}

\begin{proof}
Let \hitter be a \dyadhitter in \range. 
We prove by induction on $1\leq m \leq n$ that for a dyadic distribution \dist on \range with support size $m$, $\costt{\hitter}{\dist}=H(\dist)$.
Since $\opt{\dist} = H(\dist)$, Lemma~\ref{lem:optimal-dyadic} implies that \hitter is an optimal set of questions. 

The base case, $m = 1$, is trivial. Suppose therefore that \dist is a dyadic distribution whose support has size $m>1$. In particular, \dist is not constant, and so \hitter contains some set $S \in \dyadof{\dist}$. Let $\alpha = \dist|_S$ and $\beta = \dist|_{\overline{S}}$, and note that $\alpha,\beta$ are both dyadic. The induction hypothesis shows that $\costt{\hitter}{\alpha} = H(\alpha)$ and $\costt{\hitter}{\beta} = H(\beta)$. A decision tree which first queries $S$ and then uses the implied algorithms for $\alpha$ and $\beta$ has cost
\[
 1 + \frac{1}{2} H(\alpha) + \frac{1}{2} H(\beta) =
 h(\dist(S)) + \dist(S) H(\dist|_S) + \dist(\overline{S}) H(\dist|_{\overline{S}}) = H(\dist),
\]
using the Bernoulli chain rule; here $\dist|_S$ is the restriction of $\dist$ to the elements in $S$.

Conversely, suppose that \hitter is not a dyadic hitter, and let \dist be a non-constant dyadic distribution such that \dyadof{\dist} is disjoint from \hitter. Let $T$ be any decision tree for \dist using \hitter, and let $S$ be its first question. The cost of $T$ is at least
\[
 1 + \dist(S) H(\dist|_S) + \dist(\overline{S}) H(\dist|_{\overline{S}}) > h(\dist(S)) + \dist(S) H(\dist|_S) + \dist(\overline{S}) H(\dist|_{\overline{S}}) = H(\dist),
\]
since $\dist(S) \neq \frac{1}{2}$. Thus $\costt{\hitter}{\dist} > \opt{\dist}$, and so \hitter is not an optimal set of questions.
\end{proof}
  
\subsection{Dyadic sets as antichains} \label{sec:antichains}

There is a surprising connection between dyadic hitters and hitting sets for maximal antichains.
We start by defining the latter:

\begin{definition}
 A \emph{fibre} in $X_n$ is a subset of $2^{X_n}$ which intersects every maximal antichain in $X_n$.
\end{definition}

Fibres were defined by Lonc and Rival~\cite{LoncRival}, who also gave a simple construction, via cones:

\begin{definition}
 The \emph{cone} $\cone{S}$ of a set $S$ consists of all subsets and all supersets of $S$. 
\end{definition}

Any cone $\cone{S}$ intersects any maximal antichain $A$, since otherwise $A \cup \{S\}$ is also an antichain. By choosing $S$ of size $\lfloor n/2 \rfloor$, we obtain a fibre of size $2^{\lfloor n/2 \rfloor} + 2^{\lceil n/2 \rceil} - 1 = \Theta(2^{n/2})$. 
Our goal now is to show that every fibre is a dyadic hitter:
\begin{theorem} \label{thm:fibre}
 every fibre is a dyadic hitter.
\end{theorem}
This shows that every cone is a dyadic hitter, and allows us to give a simple algorithm for constructing an optimal decision tree using any cone.

We start with a simple technical lemma which will also be used in Section~\ref{sec:un0-ub}:

\begin{definition} \label{def:tail}
 Let $\mu$ be a dyadic distribution over $X_n$. The \emph{tail} of $\mu$ is the largest set of elements $T \subseteq X_n$ such that for some $a \geq 1$,
\begin{enumerate}[label=(\roman*)]
 \item The elements in $T$ have probabilities $2^{-a-1},2^{-a-2},\ldots,2^{-a-(|T|-1)},2^{-a-(|T|-1)}$.
 \item Every element not in $T$ has probability at least $2^{-a}$. \new{Shay: It could be good if we demonstrate the name ``tail'' by considering the tree representation of the dyadic distribution (I can do it sometimes later).}
\end{enumerate}
\end{definition}

\begin{lemma} \label{lem:small-elements}
 Suppose that $\mu$ is a non-constant dyadic distribution with non-empty tail $T$. Every set in $\dyadof{\mu}$ either contains $T$ or is disjoint from $T$.
\end{lemma}
\begin{proof}
 The proof is by induction on $|T|$. If $|T|=2$ then there exist an integer $a \geq 1$ and two elements, without loss of generality $x_1,x_2$, of probability $2^{-a-1}$, such that all other elements have probability at least $2^{-a}$. Suppose that $S \in \dyadof{\mu}$ contains exactly one of $x_1,x_2$. Then
\[
 2^{a-1} = \sum_{x_i \in S} 2^a \mu(x_i) = \sum_{x_i \in S \setminus \{x_1,x_2\}} 2^a \mu(x_i) + \frac{1}{2}.
\]
 However, the left-hand side is an integer while the right-hand side is not. We conclude that $S$ must contain either both of $x_1,x_2$ or none of them.
 
 For the induction step, let the elements in the tail $T$ of $\mu$ have probabilities $2^{-a-1},2^{-a-2},\allowbreak\ldots,\allowbreak2^{-a-(|T|-1)},2^{-a-(|T|-1)}$. Without loss of generality, suppose that $x_{n-1},x_n$ are the elements whose probability is $2^{-a-(|T|-1)}$. The same argument as before shows that every dyadic set in $\dyadof{\mu}$ must contain either both of $x_{n-1},x_n$ or neither. Form a new dyadic distribution $\nu$ on $X_{n-1}$ by merging the elements $x_{n-1},x_n$ into $x_{n-1}$, and note that $\dyadof{\mu}$ can be obtained from $\dyadof{\nu}$ by replacing $x_{n-1}$ with $x_{n-1},x_n$. The distribution $\nu$ has tail $T' = T \setminus \{x_n\}$, and so by induction, every set in $\dyadof{\nu}$ either contains $T'$ or is disjoint from $T'$. This implies that every set in $\dyadof{\mu}$ either contains $T$ or is disjoint from $T$.
\end{proof}

The first step in proving Theorem~\ref{thm:fibre} is a reduction to dyadic distributions having full support:

\begin{lemma} \label{lem:dyadic-full-support}
 A set of questions is a dyadic hitter in $X_n$ if and only if it intersects $\dyadof{\mu}$ for all non-constant full-support dyadic distributions $\mu$ on $X_n$.
\end{lemma}
\begin{proof}
 A dyadic hitter clearly intersects $\dyadof{\mu}$ for all non-constant full-support dyadic distributions on $X_n$. For the other direction, suppose that \hitter is a set of questions that intersects $\dyadof{\mu}$ for every non-constant full-support dyadic distribution $\mu$. Let $\nu$ be a non-constant dyadic distribution on $X_n$ which doesn't have full support. Let $x_{\min}$ be an element in the support of $\nu$ with minimal probability, which we denote $\nu_{\min}$. Arrange the elements in $\overline{\supp\nu}$ in some arbitrary order $x_{i_1},\ldots,x_{i_m}$. Consider the distribution $\mu$ given by:
\begin{itemize}
 \item $\mu(x_i) = \nu(x_i)$ if $x_i \in \supp \mu$ and $x_i \neq x_{\min}$.
 \item $\mu(x_{\min}) = \nu_{\min}/2$.
 \item $\mu(x_{i_j}) = \nu_{\min}/2^{j+1}$ for $j < m$.
 \item $\mu(x_{i_m}) = \nu_{\min}/2^m$.
\end{itemize}
 In short, we have replaced $\nu(x_{\min}) = \nu_{\min}$ with a tail $x_{\min},x_{i_1},\ldots,x_{i_m}$ of the same total probability. It is not hard to check that $\mu$ is a non-constant dyadic distribution having full support on $X_n$.
 
 We complete the proof by showing that \hitter intersects $\dyadof{\nu}$. By assumption, \hitter intersects $\dyadof{\mu}$, say at a set $S$. Lemma~\ref{lem:small-elements} shows that $S$ either contains all of $\{x_{\min}\} \cup \overline{\supp\nu}$, or none of these elements. In both cases, $\nu(S) = \mu(S) = 1/2$, and so \hitter intersects $\dyadof{\nu}$.
\end{proof}

We complete the proof of Theorem~\ref{thm:fibre} by showing that dyadic sets corresponding to full-support distributions are maximal antichains:

\begin{lemma} \label{lem:antichain}
Let $\mu$ be a non-constant dyadic distribution over $X_n$ with full support, and let $\dyad=\dyadof{\mu}$.
Then $\dyad$ is a maximal antichain which is closed under complementation (i.e.\ $A\in \dyad\implies X\setminus A\in \dyad$).
\end{lemma}
\begin{proof}
(i) That $\dyad$ is closed under complementation follows since if $A\in \dyad$ then $\mu(X\setminus A) = 1 - \mu(A) = 1/2$.

(ii) That $\dyad$ is an antichain follows since if $A$ strictly contains $B$ then $\mu(A) > \mu(B)$ (because $\mu$ has full support).

(iii) It remains to show that $\dyad$ is maximal.
By~(i) it suffices to show that every $B$ with $\mu(B)>1/2$ contains some $A\in \dyad$. This follows from applying Lemma~\ref{lem:neat-sum} on the probabilities of the elements in~$B$.
\end{proof}

Cones allow us to prove Theorem~\ref{thm:cone}:

\begin{proof}[Proof of Theorem~\ref{thm:cone}]
 Let $S = \{x_1,\ldots,x_{\lfloor n/2 \rfloor}\}$. The set of questions \hitter is the cone $\cone{S}$, whose size is $2^{\lfloor n/2 \rfloor} + 2^{\lceil n/2 \rceil} - 1 < 2^{\lceil n/2 \rceil + 1}$.
 
 An efficient indexing scheme for \hitter divides the index into a bit $b$, signifying whether the set is a subset of $S$ or a superset of $S$, and $\lfloor n/2 \rfloor$ bits (in the first case) or $\lceil n/2 \rceil$ bits (in the second case) for specifying the subset or superset.
 
 \smallskip
 
 To prove the other two parts, we first solve an easier question.
 Suppose that $\mu$ is a non-constant dyadic distribution whose sorted order is known. We show how to find a set in $\dyadof{\mu} \cap \hitter$ in time $O(n)$. If $\mu(S) = 1/2$ then $S \in \dyadof{\mu}$. If $\mu(S) > 1/2$, go over the elements in $S$ in non-decreasing order. According to Lemma~\ref{lem:neat-sum}, some prefix will sum to $1/2$ exactly. If $\mu(S) < 1/2$, we do the same with $\overline{S}$, and then complement the result.
 
 Suppose now that $\pi$ is a non-constant distribution. We can find a Huffman distribution $\mu$ for $\pi$ and compute the sorted order of $\pi$ in time $O(n\log n)$. The second and third part now follow as in the proof of Lemma~\ref{lem:hitterIsOptimal}.
\end{proof}

\subsection{Reduction to maximum relative density} \label{sec:un0-density} 

Our lower bound on the size of a \dyadhitter, proved in the following subsection, will be along the following lines.
For appropriate values of $n$, we describe a dyadic distribution $\mu$, all of whose splitters have a certain size $i$ or $n-i$. Moreover, only a $\rho$ fraction of sets of size $i$ split $\mu$. We then consider all possible permutations of $\mu$. Each set of size $i$ splits a $\rho$ fraction of these, and so any dyadic hitter must contain at least $1/\rho$ sets.

This lower bound argument prompts the definition of \emph{maximum relative density} (MRD), which corresponds to the parameter $\rho$ above; in the general case we will also need to optimize over~$i$. We think of the MRD as a property of dyadic sets rather than dyadic distributions; indeed, the concept of MRD makes sense for any collection of subsets of $X_n$. If a dyadic set has MRD $\rho$ then any dyadic hitter must contain at least $1/\rho$ questions, due to the argument outlined above. Conversely, using the probabilistic method we will show that roughly $1/\smallestdens{n}$ questions suffice, where $\smallestdens{n}$ is the minimum MRD of a dyadic set on \range.

\begin{definition}[Maximum relative density]\label{dfn:relativedensity}
Let \dyad be a collection of subsets of \range.
 For $0\leq i\leq n$, let
 \[\reldensof{\dyad}{i}\defeq \frac{\Bigl\lvert\bigl\{S\in \dyad : \lvert S\rvert = i\bigr\}\Big\rvert}{\binom{n}{i}}.\]
  We define the \emph{maximum relative density} (MRD) of \dyad, denoted \maxdensof{\dyad},
  as 
  \[\maxdensof{\dyad}\defeq \max_{i \in \{1,\ldots,n-1\}} \reldensof{\dyad}{i}.\]

 We define \smallestdens{n} to be the minimum of \maxdensof{\dyad} over all dyadic sets.
 That is, \smallestdens{n} is the smallest possible maximum relative density of a set of the form \dyadof{\dist}.
 \end{definition}

The following theorem shows that $\uhuf(n,0)$ is controlled by \smallestdens{n}, up to polynomial factors.
\begin{thm}\label{thm:hittersizeisdensity}
 Fix an integer $n$, and denote $M\defeq \frac{1}{\smallestdens{n}}$.
 Then
 \[ M\leq \uhuf(n,0)\leq n^2 \log n \cdot M. \]
\end{thm}

\begin{proof}
Note first that according to Lemma \ref{lem:hitterIsOptimal}, $\uhuf(n,0)$ is equal to the
minimal size of a dyadic hitter in \range, and thus it suffices to lower- and upper-bound this size.

Let \perm be a uniformly random permutation on \range. If $S$ is any set of size~$i$ then $\perm^{-1}(S)$ is a uniformly random set of size~$i$, and so
\[
 \reldensof{\dyad}{i} = \Pr_{\perm\in\sym}[\perm^{-1}(S) \in \dyad] = \Pr_{\perm\in\sym}[S \in \perm(\dyad)].
\]
(Here $\sym$ is the group of permutations of $X_n$.)

Fix a dyadic set \dyad on \range with $\maxdensof{\dyad}=\smallestdens{n}$.
The formula for $\reldensof{\dyad}{i}$ implies that for any subset $S$ of \range (of \emph{any} size),
\[\Pr_{\perm\in\sym} [S\in \perm(\dyad)]\leq \smallestdens{n}.\]

Let \hitter be a collection of subsets of \range with $|\hitter|<M$.
A union bound shows that
\[
 \Pr_{\perm\in\sym} [\hitter \cap \perm(\dyad) \neq \emptyset] \leq |\hitter| \smallestdens{n} < 1.\]
Thus, there exists a permutation $\perm$ such that $\hitter\cap \perm(\dyad) =\emptyset$. Since $\perm(\dyad)$ is also a dyadic set, this shows that \hitter is not a dyadic hitter. We deduce that any dyadic hitter must contain at least $M$ questions.

\smallskip

For the upper bound on $\uhuf(n,0)$, construct a set of subsets \hitter containing, for each $i\in \{1,\ldots,n-1\}$,
$M n \log n$ uniformly chosen sets $S\subseteq \range$ of size $i$.
We show that with positive probability, \hitter is a \dyadhitter.

Fix any dyadic set \dyad, and let $i\in\{1,\ldots,n-1\}$ be such that
$\reldensof{\dyad}{i}=\maxdensof{\dyad}\geq \smallestdens{n}$.
The probability that a random set of size $i$ doesn't belong to \dyad
is at most $1-\maxdensof{\dyad}\leq 1-\smallestdens{n}$.
Therefore the probability that \hitter is disjoint from \dyad is at most
\[(1-\smallestdens{n})^{M n \log n} \leq e^{-\smallestdens{n} M n \log n}
 = e^{-n\log n}<n^{-n}.
\]
As we show below in Claim~\ref{clm:numofdyadics}, there are at most $n^n$ non-constant dyadic distributions, and so a union bound implies that with positive probability, \hitter is indeed a \dyadhitter.
\end{proof}

In order to complete the proof of Theorem~\ref{thm:hittersizeisdensity}, we bound the number of non-constant dyadic distributions:

\begin{claim}\label{clm:numofdyadics}
 There are at most $n^n$ non-constant dyadic distributions on $X_n$.
\end{claim}
\begin{proof}
 Recall that dyadic distributions correspond to decision trees in which an element of probability $2^{-\ell}$ is a leaf at depth $\ell$. Clearly the maximal depth of a leaf is $n-1$, and so the probability of each element in a non-constant dyadic distribution is one of the $n$ values $0,2^{-1},\ldots,2^{-(n-1)}$. The claim immediately follows.
\end{proof}

Krenn and Wagner~\cite{KrennWagner} showed that the number of full-support dyadic distributions on $X_n$ is asymptotic to $\alpha \gamma^{n-1} n!$, where $\alpha \approx 0.296$ and $\gamma \approx 1.193$, implying that the number of dyadic distributions on $X_n$ is asymptotic to $\alpha e^{1/\gamma} \gamma^{n-1} n!$. Boyd~\cite{Boyd} showed that the number of monotone full-support dyadic distributions on $X_n$ is asymptotic to $\beta \lambda^n$, where $\beta \approx 0.142$ and $\lambda \approx 1.794$, implying that the number of monotone dyadic distributions on $X_n$ is asymptotic to $\beta (1+\lambda)^n$.

\medskip

The proof of Theorem~\ref{thm:hittersizeisdensity} made use of two properties of dyadic sets:
\begin{enumerate}
\item Any permutation of a dyadic set is a dyadic set.
\item There are $e^{n^{O(1)}}$ dyadic sets.	
\end{enumerate}
If $\mathcal{F}$ is any collection of subsets of $2^{X_n}$ satisfying the first property then the proof of Theorem~\ref{thm:hittersizeisdensity} generalizes to show that the minimal size $U$ of a hitting set for $\mathcal{F}$ satisfies
\[
 M \leq U \leq M n \log |\mathcal{F}|, \qquad \text{where } M = \frac{1}{\min_{D \in \mathcal{F}} \rho(D)}.
\]

\subsection{Upper bounding \smallestdenstitle{n}} \label{sec:un0-ub}

Theorem~\ref{thm:minimum-redundancy-lb} will ultimately follow from the following lemma, by way of Theorem~\ref{thm:hittersizeisdensity}:

\begin{lemma}\label{lem:densityUB}
 Fix $0 < \epsparam \leq 1/2$.
 There exists an infinite sequence of positive integers $n$ (namely, those of the form $\lfloor \frac{2^a}{2\epsparam} \rfloor$ for integer $a$) such that
 some dyadic set 
 \dyad in \range satisfies
 $\maxdensof{\dyad}\leq O(\sqrt{n}) 2^{-(h(\epsparam)-2\epsparam) n}$.
\end{lemma}

\begin{proof}
 We prove the lemma under the simplifying assumption that $1/\epsparam$ is an integer (our most important application of the lemma has $\epsparam \defeq 1/5$). Extending the argument for general $\epsparam$ is straightforward and left to the reader.

 Let $n$ be an integer of the form $\frac{2^{a}}{2\epsparam}$, for a positive integer $a$. 
 Note that for $n$ of this form, $\epsparam n = 2^{a-1}$ is a power of two.
 Let $t = \epsparam n$, and construct a dyadic distribution \dist on \range as follows:
 \begin{enumerate}
  \item For $i\in [2t-1]$, $\mu(x_i) = 2^{-a}=\frac{1}{2t}$.
  \item For $i\in [n-1] \setminus [2t-1]$, $\mu(x_i) = \mu(x_{i-1})/2 = 2^{-(a+i-2t+1)}$.
  \item $\mu(x_n) = \mu(x_{n-1})$.
 \end{enumerate}
The corresponding decision tree is obtained by taking a complete binary tree of depth $a$ and replacing one of the leaves by a ``path'' of length $n-2^a$; see Figure~\ref{fig:hard-distribution}. Alternatively, in the terminology of Definition~\ref{def:tail} we form \dist by taking the uniform distribution on $X_{2t}$ and replacing $x_{2t}$ with a tail on $x_{2t},\ldots,x_n$.

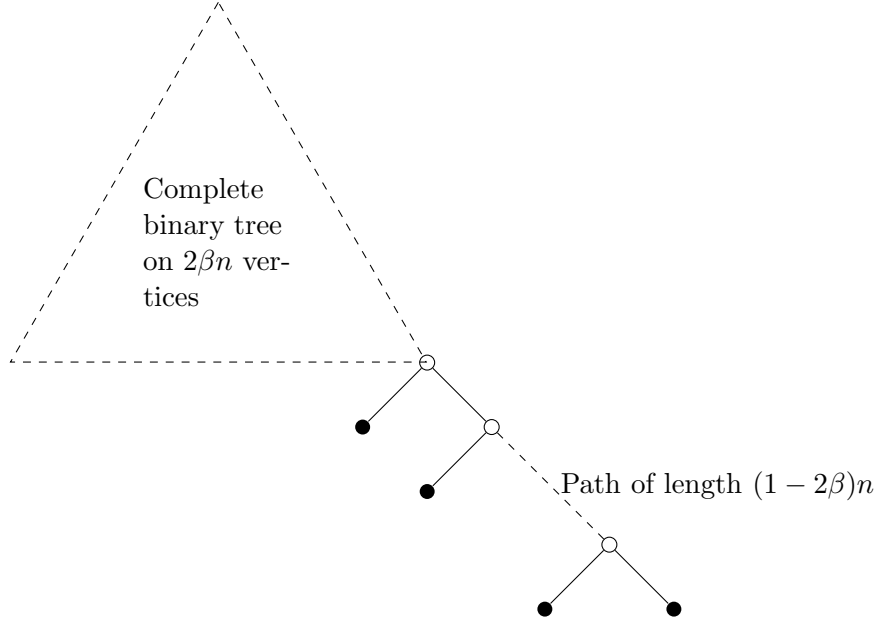
\begin{figure}

\begin{center}
\tikzset{internal/.style={circle,draw,inner sep=2pt}}
\tikzset{leaf/.style={circle,fill,inner sep=2pt}}
\begin{tikzpicture}
\node [draw,regular polygon,regular polygon sides=3,dashed,text width=2cm] (triangle) {Complete binary tree on $2\epsparam n$ vertices};
\node (leaf) [internal] at (triangle.corner 3) {};
\node (v1) [leaf,below left=1cm of leaf] {};
\node (v2) [internal,below right=1cm of leaf] {};
\node (v3) [leaf,below left=1cm of v2] {};
\node (v4) [internal,below right=2cm of v2] {};
\node (v5) [leaf,below left=1cm of v4] {};
\node (v6) [leaf,below right=1cm of v4] {};
\path
(leaf) edge (v1)
(leaf) edge (v2)
(v2) edge (v3)
(v2) edge[dashed] node[right] {Path of length $(1-2\epsparam)n$} (v4)
(v4) edge (v5)
(v4) edge (v6)
;
\end{tikzpicture}
\end{center}

\caption{The hard distribution used to prove Lemma~\ref{lem:densityUB}, in decision tree form}
\label{fig:hard-distribution}
\end{figure}
 
We claim that $\dyad\defeq \dyadof{\dist}$ contains only two types of sets:
\begin{enumerate}
 \item Subsets of size $t$ of $X_{2t-1}$.
 \item Subsets of size $n-t$ containing $t-1$ elements of $X_{2t-1}$ and all the elements $x_{2t},\ldots,x_n$.
\end{enumerate}
It is immediate that any such set $S$ is in \dyad.
On the other hand, Lemma~\ref{lem:small-elements} shows that every set $S \in \dyad$ either contains the tail $x_{2t},\ldots,x_n$ or is disjoint from it. If $S$ is disjoint from the tail then it must be of the first form, and if $S$ contains the tail then it must be of the second form.

Using the estimate $\binom{n}{\epsparam n} \geq 2^{h(\epsparam)n}/O(\sqrt{n})$ (see for example~\cite{TCS14476}), we see that 
\[\reldensof{\dyad}{t}= \reldensof{\dyad}{n-t}=\frac{\binom{2t-1}{t}}{\binom{n}{t}} \leq \frac{2^{2t}}{\binom{n}{\epsparam n}}\leq O(\sqrt{n}) \frac{2^{2t}}{2^{h(\epsparam) n}} = O(\sqrt{n}) 2^{(2\epsparam - h(\epsparam)) n}.\]
For $i\in \{1,\ldots,n-1\} \setminus \set{t,n-t}$ we have $\reldensof{\dyad}{i}=0$.
Thus indeed
\[\maxdensof{\dyad}\leq  O(\sqrt{n}) 2^{(2\epsparam - h(\epsparam)) n}. \qedhere \]
\end{proof}

Theorem~\ref{thm:minimum-redundancy-lb} can now be easily derived.
The first step is determining the optimal value of $\epsparam$:
\begin{claim}\label{claim:maxOfFunc}
We have
 \[ \max_{\epsparam \in [0,1]} 2^{h(\epsparam)-2 \epsparam} = 1.25, \]
and the maximum is attained (uniquely) at $\epsparam=1/5$.
 \end{claim}
\begin{proof}
 Let $f(\epsparam) = h(\epsparam)-2\epsparam$.
 Calculation shows that the derivative $f'(\epsparam)$ is equal to 
 \[f'(\epsparam) = \log \left(\frac{1-\epsparam}{\epsparam}\right) -2,\]
 which is decreasing for $0<\epsparam<1$ and vanishes at $\epsparam =1/5$.
 Thus $f(\epsparam)$ achieves a unique maximum
 over $\epsparam\in (0,1)$ at $\epsparam=1/5$,
 where
 \[2^{f(1/5)} = 2^{h(1/5)-2\cdot 1/5} =1.25. \qedhere\]
\end{proof}

\paragraph{Proof of Theorem \ref{thm:minimum-redundancy-lb}:}
\begin{proof}
Let $\epsparam\defeq 1/5$. Claim~\ref{claim:maxOfFunc} shows that $2^{-(2\epsparam-h(\epsparam))}=1.25$.
Fix any $n$ of the form $n=\frac{2^{a}}{2\epsparam}$ for a positive integer $a$. 
It follows from Lemma~\ref{lem:densityUB} together with the first inequality in Theorem~\ref{thm:hittersizeisdensity}
that $\uhuf(n,0)\geq 1.25^n/O(\sqrt{n})$.

A general $n$ can be written in the form $n = \frac{2^a}{2\beta}$ for a positive integer $a$ and $1/4 \leq \beta \leq 1/2$.
Lemma~\ref{lem:densityUB} and Theorem~\ref{thm:hittersizeisdensity} show that for any integer $\ell \geq 0$,
\[
 \uhuf(n,0) \geq 2^{[h(\beta/2^\ell) - 2\beta/2^\ell]n}/O(\sqrt{n}).
\]
Calculation shows that when $\beta \leq \beta_0 \approx 0.27052059413118146$, this is maximized at $\ell = 0$, and otherwise this is maximized at $\ell = 1$.
Denote the resulting lower bound by $L(\beta)^n/O(\sqrt{n})$, the minimum of $L(\beta)$ is attained at $\beta_0$, at which point its value is $L(\beta_0) \approx 1.23214280723432$.
\end{proof}

\subsection{Lower bounding \smallestdenstitle{n}} \label{sec:un0-lb}

We will derive Theorem~\ref{thm:minimum-redundancy-ub} from the following lemma:
\begin{lemma}\label{lem:min-redundancy-ubhelper}
 For every non-constant dyadic distribution \dist
 there exists $0<\epsparam<1$ such that 
 \[\maxdensof{\dyadof{\dist}}\geq \frac{ 2^{(2\epsparam - h(\epsparam)) n}}{O(\sqrt n)^{O(\log n)}} = 2^{(2\epsparam - h(\epsparam)) n-o(n)}.\]
 \end{lemma}
\begin{proof}
 Assume without loss of generality that the probabilities in \dist are non-increasing:
 \[ \dist_1 \geq \dist_2 \geq \cdots \geq \dist_n. \]
 
 The idea is to find a partition of $X_n$ of the form
 \[
  X_n = \bigcup_{i=1}^\gamma (D_i \cup E_i)
 \]
 which satisfies the following properties:
\begin{enumerate}
 \item $D_i$ consists of elements having the same probability $p_i$.
 \item If $D_i$ has an even number of elements then $E_i = \emptyset$.
 \item If $D_i$ has an odd number of elements then $\dist(E_i) = p_i$.
 \item $\gamma = O(\log n)$. (In fact, $\gamma = o(n/\log n)$ would suffice.)
\end{enumerate}
 We will show later how to construct such a partition.
 
 The conditions imply that $\dist(D_i \cup E_i)$ is an even integer multiple of $p_i$, say $\dist(D_i \cup E_i) = 2c_ip_i$. It is not hard to check that $c_i = \lceil |D_i|/2 \rceil$.
 
 Given such a partition, we show how to lower bound the maximum relative density of $\dyadof{\dist}$. If $S_i \subseteq D_i$ is a set of size $c_i$ for each $i \in [\gamma]$ then the set $S = \bigcup_i S_i$ splits $\dist$:
\[
 \dist(S) = \sum_{i=1}^\gamma c_i p_i = \frac{1}{2} \sum_{i=1}^\gamma \dist(D_i \cup E_i) = \frac{1}{2}.
\]
 Defining $c = \sum_{i=1}^\gamma c_i$, we see that each such set $S$ contains $c$ elements, and the number of such sets is
\[
 \prod_{i=1}^\gamma \binom{|D_i|}{c_i} \geq \prod_{i=1}^\gamma \frac{2^{2c_i}}{O(\sqrt{n})} = \frac{2^{2c}}{O(\sqrt{n})^{O(\log n)}},
\]
 using the estimate
\[
 \binom{m}{\lceil m/2 \rceil} = \Theta\left(\frac{2^{2\lceil m/2 \rceil}}{\sqrt{m}}\right),
\]
 which follows from Stirling's approximation.

 In order to obtain an estimate on the maximum relative density of $\dyadof{\dist}$, we use the following folklore upper bound\footnote{Here is a quick proof: Let $Y$ be a uniformly random subset of $X_n$ of size $c$, and let $Y_i$ indicate the event $x_i \in Y$. Then $\log \binom{n}{c} = H(Y) \leq nH(Y_1) = nh(c/n)$.} on $\binom{n}{c}$:
\[
 \binom{n}{c} \leq 2^{h(c/n)n}. 
\]
 We conclude that the maximum relative density of $\dyadof{\dist}$ is at least
\[
 \maxdensof{\dist} \geq \reldensof{\dist}{c} \geq
 \frac{\prod_{i=1}^\gamma \binom{|D_i|}{c_i}}{\binom{n}{c}} \geq
 \frac{2^{2c-h(c/n)n}}{O(\sqrt{n})^{O(\log n)}}.
\]
 To obtain the expression in the statement of the lemma, take $\epsparam \defeq c/n$.
 
 \medskip
 
 We now show how to construct the partition of $X_n$. We first explain the idea behind the construction, and then provide full details; the reader who is interested only in the construction itself can skip ahead.
 
 \paragraph{Proof idea} Let $q_1,\ldots,q_\gamma$ be the different probabilities of elements in \dist. We would like to put all elements of probability $q_i$ in the set $D_i$, but there are two difficulties:
\begin{enumerate}
\item There might be an odd number of elements whose probability is $q_i$.
\item There might be too many distinct probabilities, that is, $\gamma$ could be too large. (We need $\gamma = o(n/\log n)$ for the argument to work.)
\end{enumerate}

 The second difficulty is easy to solve: we let $D_1 = \{ x_1 \}$, and use Lemma~\ref{lem:neat-sum} to find an index $\ell$ such that $\dist(E_1) \defeq \dist(\{x_\ell,\ldots,x_n\}) = \dist_1$. A simple argument shows that all remaining elements have probability at least $\dist_1/n$, and so the number of remaining distinct probabilities is $O(\log n)$. (The reader should observe the resemblance between $E_1$ and the tail of the hard distribution constructed in Lemma~\ref{lem:densityUB}.)
 
 Lemma~\ref{lem:neat-sum} also allows us to resolve the first difficulty. The idea is as follows. Suppose that the current set under construction, $D_i$, has an odd number of elements, each of probability $q_i$. We use Lemma~\ref{lem:neat-sum} to find a set of elements whose total probability is $q_i$, and put them in $E_i$.
 
 \paragraph{Detailed proof}
 
 Let $N$ be the maximal index such that $\dist_N > 0$. Since $\dist$ is non-constant, $\dist_1 \leq 1/2$, and so Lemma~\ref{lem:neat-sum} proves the existence of an index $M$ such that $\dist(\{x_{M+1},\ldots,x_N\}) = \dist_1$ (we use the \emph{furthermore} part of the lemma, and $M = \ell-1$). We take
\[
 D_1 \defeq \{ x_1 \}, \quad E_1 \defeq \{ x_{M+1}, \ldots, x_n \}.
\]
 Thus $\dist(D_1) = \dist(E_1) = \dist_1$, and so $\dist(\{x_2,\ldots,x_M\}) = 1-2\dist_1$ (possibly $M=1$, in which case the construction is complete).
 
 By construction $n\dist_M > \dist(E_1) = \dist_1$, and so $\dist_M < \dist_1/n$. In particular, the number of distinct probabilities among $\dist_2,\ldots,\dist_M$ is at most $\log n$. This will guarantee that $\gamma \leq \log n + 1$, as will be evident from the construction.
 
 The construction now proceeds in steps. At step $i$, we construct the sets $D_i$ and $E_i$, given the set of available elements $\{x_{\alpha_i},\ldots,x_M\}$, where possibly $\alpha_i = M+1$; in the latter case, we have completed the construction. We will maintain the invariant that $\dist(\{x_{\alpha_i},\ldots,x_M\})$ is an even multiple of $\dist_{\alpha_i}$;
 initially $\alpha_2 \defeq 2$, and $\dist(\{x_{\alpha_i},\ldots,x_M\}) = (1/\dist_1-2)\dist_1$ is indeed an even multiple of $\dist_2$.
 
 Let $\beta_i$ be the maximal index such that $\dist_{\beta_i} = \dist_{\alpha_i}$ (possibly $\beta_i = \alpha_i$). We define
\[
 D_i \defeq \{x_{\alpha_i},\ldots,x_{\beta_i}\}.
\]

 Suppose first that $|D_i|$ is even. In this case we define $E_i \defeq \emptyset$, and $\alpha_{i+1} \defeq \beta_i+1$. Note that
\[
 \dist(\{x_{\alpha_{i+1}},\ldots,x_M\}) =
 \dist(\{x_{\alpha_i},\ldots,x_M\}) - |D_i| \dist_{\alpha_i},
\]
 and so the invariant is maintained.
 
 Suppose next that $|D_i|$ is odd. In this case $\dist(\{x_{\beta_i+1},\ldots,x_M\}) \geq x_{\alpha_i}$, since $\dist(\{x_{\beta_i+1},\ldots,x_M\})$ is an odd multiple of $\dist_{\alpha_i}$. Therefore we can use Lemma~\ref{lem:neat-sum} to find an index $\gamma_i$ such that $\dist(\{x_{\beta_i+1},\ldots,x_{\gamma_i}\}) = \dist_{\alpha_i}$. We take
\[
 E_i \defeq \{x_{\beta_i+1},\ldots,x_{\gamma_i}\}
\]
 and $\alpha_{i+1} \defeq \gamma_i+1$. Note that
\[
 \dist(\{x_{\alpha_{i+1}},\ldots,x_M\}) =
 \dist(\{x_{\alpha_i},\ldots,x_M\}) - (|D_i|+1) \dist_{\alpha_i},
\]
 and so the invariant is maintained.
 
 The construction eventually terminates, say after step $\gamma$. The construction ensures that $\dist_{\alpha_2} > \dist_{\alpha_3} > \cdots > \dist_{\alpha_\gamma}$. Since there are at most $\log n$ distinct probabilities among the elements $\{x_{\alpha_2},\ldots,x_M\}$, $\gamma \leq \log n + 1$, completing the proof.
\end{proof}

Theorem \ref{thm:minimum-redundancy-ub} follows immediately from the second inequality in Theorem \ref{thm:hittersizeisdensity}
together with the following lemma:
\begin{lemma}\label{lem:densityLB}
 Fix an integer $n$ and let \dyad be a dyadic set in \range.
 Then
 \[\maxdensof{\dyad}\geq 1.25^{-n-o(n)},\]
 and thus
 \[\smallestdens{n}\geq 1.25^{-n-o(n)}.\]
 
\end{lemma}
\begin{proof}
Fix a dyadic set \dyad in \range.
Lemma~\ref{lem:min-redundancy-ubhelper} implies that there exists $0<\epsparam<1$ such that
$\maxdensof{\dyad}\geq 2^{(2\epsparam - h(\epsparam)) n-o(n)}$.
Using Claim \ref{claim:maxOfFunc} we have $2^{2\epsparam - h(\epsparam)}\geq \frac{4}{5}$,
and so 
\[\maxdensof{\dyad}\geq (4/5)^n\cdot 2^{-o(n)}= 1.25^{-n-o(n)}. \qedhere \]
 \end{proof}

\section{Combinatorial benchmark with prolixity} \label{sec:prolixity}

In the previous section we studied the minimum size of a set $\cQ$ of questions with the property that
for every distribution, there is an optimal decision tree using only questions from $\cQ$.
In this section we relax this requirement by allowing the cost to be
slightly worse than the optimal cost.

More formally, recall that $\uhuf(n,r)$ is the minimum size of a set of questions $\cQ$
such that for every distribution $\pi$ there exists a decision tree that uses 
only questions from $\cQ$ with cost at most $\opt{\pi} + r$.

In a sense, $\uhuf(n,r)$ is an extension of $\uent(n,r)$ for $r\in(0,1)$:
indeed, $\uent(n,r)$ is not defined for $r<1$ since for some distributions $\pi$ there is no decision tree with cost less than $H(\pi)+1$ (see Section~\ref{sec:comparison-equality}).
Moreover, $\opt{\pi}$, which is the benchmark used by $\uhuf(n,r)$, is precisely the optimal cost, 
whereas $H(\pi)$, the benchmark used by $\uent(n,r)$ is a convex surrogate
of $\opt{\pi}$.

We focus here on the range $r\in(0,1)$.
We prove the following bounds on $\uhuf(n,r)$, 
establishing that $\uhuf(n,r) \approx (r\cdot n)^{\Theta(1/r)}$.

\begin{main-theorem}\label{thm:min-redundancy-approx}
For all $r\in(0,1)$, and for all $n>1/r$:
\[\frac{1}{n}(r\cdot n)^{\frac{1}{4r}}\leq \uhuf(n,r) \leq n^2(3r\cdot n)^{\frac{16}{r}}. \]
\end{main-theorem}

As a corollary, we get that the threshold of exponentiality is $1 / n$:

\begin{main-corollary} \label{cor:minimum-redundancy-approx-exp}
 If $r = \omega(1/n)$ then $\uhuf(n,r) = 2^{o(n)}$.
 
 Conversely, if $r = O(1/n)$ then $\uhuf(n,r) = 2^{\Omega(n)}$.
\end{main-corollary}

For larger $r$, the following theorem is a simple corollary of Theorem~\ref{thm:redundancy-r-bounds} and the bound $\uhuf(n,r) \leq \uent(n,r) \leq \uhuf(n,r-1)$:

\begin{main-theorem} \label{thm:prolixity-r-bounds}
For every $r \geq 1$ and $n\in\mathbb{N}$,
\[
 \frac{1}{e}\lfloor r+1 \rfloor n^{1/\lfloor r+1 \rfloor} \leq \uhuf(n,r) \leq 2 \lfloor r \rfloor n^{1/\lfloor r \rfloor}.
\]
\end{main-theorem}

\medskip

Theorem~\ref{thm:min-redundancy-approx} is implied by the following lower and upper bounds,
which provide better bounds when $r\in (0,1)$ is a negative power of~$2$.

\begin{main-theorem}[Lower bound] \label{thm:minimum-redundancy-approx-lb}

For every $r$ of the form $1/2^k$, where $k\geq 1$ is an integer, and $n>2^k$:
\[
\uhuf(n,r) \geq (r\cdot n)^{\frac{1}{2r}-1}.
\]
\end{main-theorem}

\begin{main-theorem}[Upper bound] \label{thm:minimum-redundancy-approx-ub}
For every $r$ of the form $4/2^k$, where $k\geq 3$ is an integer, and $n>2^k$:
\[
 \uhuf(n,r+ r^2) \leq n^2 \bigl(\frac{3e}{4}r\cdot n\bigr)^{\frac{4}{r}}.
\]
\end{main-theorem}

These results imply Theorem~\ref{thm:min-redundancy-approx}, due to the monotonicity of $\uhuf(n,r)$, as follows.

Let $r\in(0,1)$.
For the \emph{lower bound}, pick the smallest $t\geq r$ of the form $1/2^k$. 
Note that $t\leq 2r$, and thus:
\[\uhuf(n,r) \geq \uhuf(n,t)  \geq (t\cdot n)^{\frac{1}{2t}-1} \geq (r\cdot n)^{\frac{1}{4r}-1}\geq \frac{1}{n}(r\cdot n)^{\frac{1}{4r}}.\]

For the \emph{upper bound}, pick the largest $t$ of the form $4/2^k$, $k\geq 3$ such that $t + t^2\leq r$.
Note that $t\geq r/4$ (since $s=r/2$ satisfies $s+s^2 \leq 2s \leq r$), and thus
\[\uhuf(n,r) \leq \uhuf(n, t+t^2) \leq  n^2 \bigl(\frac{3e}{4}t\cdot n\bigr)^{\frac{4}{t}}\leq n^2(3r\cdot n)^{\frac{16}{r}}.\]

\subsection{Lower bound} \label{sec:minimum-redundancy-approx-lb}

Pick a sufficiently small $\delta>0$ (as we will soon see, $\delta < r^2$ suffices), and consider a distribution $\mu$ with $2^k-1$ ``heavy'' elements (this many elements exist since $n>1/r$), each of probability $\frac{1-\delta}{2^k-1}$, and $n-(2^k-1)$ ``light'' elements with total probability of $\delta$. Recall that a decision tree is \emph{$r$-optimal} if its cost is at most $\opt{\mu}+r$.
The proof proceeds by showing that if $T$ is an $r$-optimal tree, then the first question in $T$ has the following properties:
\begin{enumerate}[label=(\roman*)]
\item it separates the heavy elements to two sets of almost equal sizes ($2^{k-1}$ and $2^{k-1}-1$), and
\item it does not distinguish between the light elements. 
\end{enumerate}
The result then follows since there are $\binom{n}{2^k-1}$ such distributions $\sigma$ (the number of  ways to choose the light elements), and each question can serve as a first question to at most $\binom{n-(2^{k-1}-1)}{2^{k-1}}$ of them.

To establish these properties, we first prove a more general result (cf.\ Lemma~\ref{lem:small-elements}):

\begin{lemma}\label{l324}
Let $\mu$ be a distribution over a finite set $X$, and let $A\subseteq X$ be
such that for every $x\notin A$, $\mu\bigl(\{x\}\bigr) > \mu(A) + \eps$.
Then every decision tree $T$ which is $\eps$-optimal with respect to $\mu$ has a subtree $T'$
whose set of leaves is $A$.
\end{lemma}

\begin{proof}
By induction on $\lvert A\rvert$. 
The case $\lvert A\rvert =1$ follows since any leaf is a subtree.
Assume $\lvert A\rvert > 1$. Let $T$ be a decision tree which is $\eps$-optimal with respect to $\mu$.
Let $x,y$ be two siblings of maximal depth.
Note that it suffices to show that $x,y\in A$,
since then, merging $x,y$ to a new element $z$
with $\mu(\{z\}) = \mu(\{x\}) + \mu(\{y\})$ and applying
the induction hypothesis yields that $A\cup\{z\}\setminus\{x,y\}$
is the set of leaves of a subtree of $T$ with $x,y$ removed. 
This finishes the proof since $x,y$ are the children of $z$.

It remains to show that $x,y\in A$. 
Let $d$ denote the depth of $x$ and $y$.
Assume toward contradiction that $x\notin A$. 
Pick $a',a''\in A$, with depths $d',d''$ (this is possible since $\lvert A\rvert > 1$).
If $d'<d$ or $d''<d$ then replacing $a'$ with $x$ or $a''$ with $x$ improves
the cost of $T$ by more than $\eps$, contradicting its optimality.
Therefore, it must be that $d'=d''=d$, and we perform the following transformation
(see Figure~\ref{fig:transformation}):
the parent of $x$ and $y$ becomes a leaf with label $x$ (decreasing the depth of $x$ by 1),
$y$ takes the place of $a'$ (the depth of $y$ does not change), and $a''$ becomes an internal
node with two children labeled by $a',a''$ (increasing the depths of $a',a''$ by 1).
Since $\mu\bigl(\{x\}\bigr) - \mu\bigl(\{a',a''\}\bigr) > \eps$, 
this transformation improves the cost of $T$ by more than $\eps$, contradicting its $\epsilon$-optimality.
\end{proof}

\begin{figure}
\begin{subfigure}[t]{.5\linewidth}
\centering
\begin{tikzpicture}[baseline=(root.base),
level 1/.style={sibling distance=4cm},
level 2/.style={sibling distance=2.5cm}]
\node[circle, draw] (root) {}
    child{
    node[circle, draw, fill=black] {} edge from parent[dashed]
    	child{
	node[circle, draw, solid] {$x$} edge from parent[solid]
	     }
	child{
	node[circle, draw, solid] {$y$} edge from parent[solid]
	     }
        }
    child{ 
    node[circle, draw] {} edge from parent[dashed]
    	child{
	node[circle, draw, solid] {$a'$} 
	     }
	child{
	node[circle, draw, solid] {$a''$} 
	     }
        };
\end{tikzpicture}
\caption{Original tree}
\end{subfigure}
\begin{subfigure}[t]{.5\linewidth}
\centering
\begin{tikzpicture}[baseline=(root.base),
level 1/.style={sibling distance=2.5cm},
level 2/.style={sibling distance=2.5cm}
level 2/.style={sibling distance=2.5cm}]
\node[circle, draw] (root) {}
    child{node[circle, draw] {$x$} edge from parent[dashed]
        }
    child{ 
    node[circle, draw] {} edge from parent[dashed]
    	child{
	node[circle, draw, solid] {$y$} 
	     }
	child{
	node[circle, draw, fill=black, solid] {} 
		child{
		node[circle, draw, solid] {$a'$} edge from parent[solid]
	     	      }
		child{
	        node[circle, draw, solid] {$a''$} edge from parent[solid]
	              }
	     }
        };
\end{tikzpicture}
\caption{Transformed tree}
\end{subfigure}
\caption{The transformation in Lemma~\ref{l324}.
The cost decreases by $\mu\bigl(\{x\}\bigr)-\mu\bigl(\{a',a''\}\bigr)>\eps$.}\label{fig:transformation}
\end{figure}

\begin{corollary}\label{c325}
Let $\mu$ be a distribution over $X$, and let $A\subseteq X$ be
such that for every $x\notin A$, $\mu\bigl(\{x\}\bigr)  >  \mu(A)$.
Then every optimal tree $T$ with respect to $\mu$ has a subtree $T'$
whose set of leaves is $A$.
\end{corollary}

Property~(ii) follows from Lemma~\ref{l324}, which implies that if $\delta$ is sufficiently small then all light elements are clustered together as the leaves of some subtree. Indeed, by Lemma~\ref{l324}, 
this happens if the probability of a single heavy element (which is $\frac{1-\delta}{2^k-1}$) exceeds the total probability of all light elements (which is $\delta$) by at least $r$. A simple calculation shows that setting $\delta$ smaller than $r^2$ suffices.

We summarize this in the following claim:
\begin{claim}[light elements]
Every $r$-optimal tree has a subtree whose set of leaves is the set of light elements.
\end{claim}

The next claim concerns the other property:
\begin{claim}[heavy elements]
In every $r$-optimal decision tree, the first question partitions the heavy elements into a set of size $2^{k-1}$ and a set of size $2^{k-1}-1$.
\end{claim}
\begin{proof}
When $k = 2$, it suffices to prove that an $r$-optimal decision tree cannot have a first question which separates the heavy elements from the light elements. Indeed, the heavy elements in such a tree reside at depths $2,3,3$. Exchanging one of the heavy elements at depth~$2$ with the subtree consisting of all light elements (which is at depth~$1$) decreases the cost by $\frac{1-\delta}{2^k-1} - \delta > r$, showing that the tree wasn't $r$-optimal.

Suppose that some $r$-optimal decision tree $T$ contradicts the statement of the claim, for some $k \geq 3$. The first question in $T$ leads to two subtrees $T_1,T_2$, one of which (say $T_1$) contains at least $2^{k-1}+1$ heavy elements, and the other (say $T_2$) contain at most $2^{k-1}-2$. One of the subtrees also contains a subtree $T'$ whose leaves are all the light elements. For the sake of the argument, we replace the subtree $T'$ with a new element $y$.

We claim that $T_1$ contains an internal node $v$ at depth $D(v) \geq k-1$ which has at least two heavy descendants. To see this, first remove $y$ if it is present in $T_1$, by replacing its parent by its sibling. The possibly modified tree $T'_1$ contains at least $2^{k-1}+1$ leaves, and in particular some leaf at depth at least $k$. Its parent $v$ has depth at least $k-1$ and at least two heavy descendants, in both $T'_1$ and $T_1$.

In contrast, $T_2$ contains at least two leaves (since $2^{k-1}-2 \geq 2$), and the two shallowest ones must have depth at most $k-2$. At least one of these is some heavy element $x_\ell$.

Exchanging~$v$ and $x_\ell$ results in a tree $T^*$ whose cost $c(T^*)$ is at most
\[
 c(T^*) \leq c(T) + (D(v) - D(x_\ell))(2 - 1)\frac{1-\delta}{2^k-1} \leq c(T^*) - \frac{1-\delta}{2^k-1} < c(T) - r,
\]
contradicting the assumption that $T$ is $r$-optimal. (That $\frac{1-\delta}{2^k-1} > r$ follows from the earlier assumption $\frac{1-\delta}{2^k-1} > \delta + r$.)
\end{proof}
By the above claims, there are two types of first questions for $\mu$, 
depending on which of the two subtrees of the root contains the light elements:
\begin{itemize}
\item Type 1: questions that split the elements into a part with $2^{k-1}$ elements, and a part with $n-2^{k-1}$ elements. 
\item Type 2: questions that split the elements into a part with $2^{k-1}-1$ elements, and a part with $n-(2^{k-1}-1)$ elements.
\end{itemize}
If we identify a question with its smaller part  (i.e.\ the part of size $2^{k-1}$ or the part of size $2^{k-1}-1$), we deduce that any set of questions with redundancy $r$ must contain a family $\cF$ such that (i) every set in $\cF$ has size $2^{k-1}$ or $2^{k-1}-1$, and (ii) for every set of size $n-(2^k-1)$, there exists some set in $\cF$ that is disjoint from it. It remains to show that any such family $\cF$ is large.

Indeed, there are $\binom{n}{2^k-1}$ sets of size $n-(2^k-1)$, and since every set in $\cF$ has size at least $2^{k-1}-1$,
it is disjoint from at most $\binom{n-(2^{k-1}-1)}{n-(2^k-1)}=\binom{n-(2^{k-1}-1)}{2^{k-1}}$ of them.
Thus
\[
|\cF| \geq \frac{\binom{n}{2^k-1}}{\binom{n-(2^{k-1}-1)}{2^{k-1}}}
=
\frac{n(n-1)\cdots(n-(2^{k-1}-1)+1)}{(2^k-1)(2^k-2)\cdots(2^{k-1}+1)}
\geq
\bigl(\frac{n}{2^k}\bigr)^{2^{k-1}-1}
=
(r\cdot n)^{\frac{1}{2r}-1}.
\]

\subsection{Upper bound} \label{sec:minimum-redundancy-approx-ub}
\paragraph{The set of questions.}
In order to describe the set of queries
it is convenient to assign a cyclic order on $X_n$: $x_1 \prec x_2 \prec \cdots \prec x_n \prec x_1 \prec \cdots$.
The set of questions $\cQ$ consists of all cyclic intervals, 
with up to $2^k$ elements added or removed.
Since $r = 4 \cdot 2^{-k}$, the number of questions is plainly at most
\[
 n^2 \binom{n}{2^k} 3^{2^k} \leq n^2 \bigl(\frac{3e}{4}r\cdot n\bigr)^{\frac{4}{r}},
\]
using the inequality $\binom{n}{d}\leq\bigl(\frac{en}{d}\bigr)^d$.

\paragraph{High level of the proof.}
Let $\pi$ be an arbitrary distribution on $X_n$, and let $r\in(0,1)$ be of the form $4\cdot 2^{-k}$, with $k\geq 3$. Let $\mu$ be a Huffman distribution for $\pi$; we remind the reader that $\mu$ is a dyadic distribution corresponding to some optimal decision tree for $\pi$. We construct a decision tree $T$ that uses only queries from $\cQ$, with cost 
\[ T(\pi)\leq \opt{\pi} + r + r^2 = \sum_{x \in X_n}{\pi(x)\log\frac{1}{\mu(x)}} + r + r^2. \]
The construction is randomized: we describe a randomized decision tree $T_R$ (`$R$' denotes the randomness that determines the tree) which uses queries from $\cQ$ and has the property that for every $x\in X_n$, the expected number of queries $T_R$ uses to find $x$ satisfies the inequality
\begin{equation}\label{eq4}
\E_R[T_R(x)] \leq \log \frac{1}{\mu(x)} + r + r^2,
\end{equation}
where $T_R(x)$ is the depth of $x$.
This implies the existence of a deterministic tree with cost
$\opt{\mu} + r + r^2$: indeed, when $x\sim\mu$, the expected cost of $T_R$
is 
\[ \E_{\substack{x \sim \pi; R}}[T_R(x)]\leq \sum_{x \in X_n}{\pi(x)\Bigl(\frac{1}{\mu(x)} + r + r^2\Bigr)} = \opt{\pi} + r + r^2. \]
Since the randomness of the tree is independent from the randomness of $\pi$,
it follows that there is a choice of $R$ such that the cost of the (deterministic) decision tree $T_R$
is at most $\opt{\pi} + r + r^2$.

\paragraph{The randomized decision tree.}
The randomized decision tree maintains a dyadic \emph{sub-distribution} $\mu^{(i)}$ that is being updated after each query. A \emph{dyadic sub-distribution} is a measure on $X_n$ such that (i) $\mu^{(i)}(x)$ is either 0 or a power of 2, and (ii) $\mu^{(i)}(X_n)=\sum_{x\in X_n}\mu^{(i)}(x)\leq 1$.
A natural interpretation of $\mu^{(i)}(x)$ 
is as a dyadic sub-estimate of the probability that 
$x$ is the \unknown element, conditioned on the answers to the first $i$ queries.
The analysis hinges on the following properties:
\begin{enumerate}
\item $\mu^{(0)}=\mu$, 
\item $\mu^{(i)}(x)\in\bigl\{2\mu^{(i-1)}(x),\mu^{(i-1)}(x),0\bigr\}$ for all $x\in X_n$, 
\item if $x$ is the \unknown element then almost always $\mu^{(i)}(x)$ is doubled; 
that is, $\mu^{(i)}(x)>0$ for all $i$, 
and the expected number of $i$'s for which $\mu^{(i)}(x)=\mu^{(i-1)}(x)$ is at most $r+r^2$.
\end{enumerate}
These properties imply~\eqref{eq4}, which implies Theorem~\ref{thm:minimum-redundancy-approx-ub}.

Next, we describe the randomized decision tree and establish these properties.

The algorithm distinguishes between {light} and {heavy} elements.
An element $x\in X_n$ is \emph{light} if $\mu^{(i)}(x) < 2^{-k}$.
Otherwise it is \emph{heavy}. The algorithm is based on the following win-win-win situation:

(i) If the total mass of the heavy elements is at least $1/2$
then by Lemma~\ref{lem:neat-sum}, there is a set $I$ of heavy elements whose mass is exactly $1/2$.
Since the number of heavy elements is at most $2^k$, the algorithm can ask whether $x \in I$ and recurse by doubling
the sub-probabilities of the elements that are consistent with the answer (and setting the others to zero).

(ii) Otherwise, the mass of the heavy elements is less than $1/2$.
If the mass of the light elements is also less than $1/2$ 
(this could happen since $\mu^{(i)}$ is a sub-distribution),
then we ask whether $x$ is a heavy element or a light element,
and accordingly recurse with either the heavy or the light elements, with their
sub-probabilities doubled (in this case the ``true'' probabilities conditioned on the answers
 become larger than the sub-probabilities).

(iii) The final case is when the mass of the light elements is larger than $1/2$.
In this case we query a random cyclic interval of light elements of mass $\approx 1/2$,
and recurse; there are two light elements in the recursion whose sub-probability
is not doubled (the probabilities of the rest are doubled).

Elements whose probability is not doubled occur only in case~(iii).

\paragraph{The randomized decision tree: formal description.}

The algorithm gets as input a subset $y_1,\ldots,y_m$ of $X_n$ whose order is induced by that of $X_n$, and a dyadic sub-distribution $q_1,\ldots,q_m$. Initially, the input is $x_1,\ldots,x_n$, and $q_i = \mu_i$.

We say that an element is \emph{heavy} is $q_i \geq 2^{-k}$; otherwise it is \emph{light}. There are at most $2^k$ heavy elements. The questions asked by the algorithm are cyclic intervals in $y_1,\ldots,y_m$, with some heavy elements added or removed. Since each cyclic interval in $y_1,\ldots,y_m$ corresponds to a (not necessarily unique) cyclic interval in $X_n$ (possibly including elements outside of $y_1,\ldots,y_m$), these questions belong to $\cQ$.

\begin{algorithm-description}{$T_R$}
\begin{enumerate}
\item If $m = 1$, return $y_1$. Otherwise, continue to Step~2.
\item If the total mass of heavy elements is at least $1/2$ then find (using Lemma~\ref{lem:neat-sum}) a subset $I$ whose mass is exactly $1/2$, and ask whether $x \in I$. Recurse with either $\{ 2q_i : y_i \in I \}$ or $\{ 2q_i : y_i \notin I \}$, according to the answer. Otherwise, continue to Step~3.
\item Let $S$ be the set of all light elements, and let $\sigma$ be their total mass. If $\sigma \leq 1/2$ then ask whether $x \in S$, and recurse with either $\{ 2q_i : y_i \in S \}$ or $\{ 2q_i : y_i \notin S \}$, according to the answer. Otherwise, continue to Step~4.
\item Arrange all light elements according to their cyclic order on a circle of circumference $\sigma$, by assigning each light element $x_i$ an arc $A_i$ of length $q_i$ of the circle.  Pick an arc of length $1/2$ uniformly at random (e.g.\ by picking uniformly a point on the circle and taking an arc of length $1/2$ directed clockwise from it), which we call the \emph{\halfcircle}. Let $K \subseteq S$ consist of all light elements whose midpoints are contained in the \halfcircle, and let $B$ consist of the light elements whose arcs are cut by the boundary of the \halfcircle (so $|B| \leq 2$); we call these elements \emph{boundary elements}. Ask whether $x \in K$; note that $K$ is a cyclic interval in $y_1,\ldots,y_m$ with some heavy elements removed.

If $x \in K$, recurse with $\{2q_i : y_i \in K \setminus B\} \cup \{q_i : y_i \in K \cap B\}$. The sum of these dyadic probabilities is at most~$1$ since the \halfcircle contains at least $q_i/2$ of the arc $A_i$ for each $y_i \in K \cap B$.

If $x \notin K$, recurse with $\{2q_i : y_i \in \overline{K} \setminus B \} \cup \{q_i : y_i \in \overline{K} \cap B\}$. As in the preceding case, the total mass of light elements in the recursion is at most $2(\sigma - 1/2)$ (since the complement of the \halfcircle contains at least $q_i/2$ of the arc $A_i$ for each $y_i \in \overline{K} \cap B$), and the total mass of heavy elements is $2(1-\sigma)$, for a total of at most $(2\sigma-1)+(2-2\sigma) = 1$. \eofahere
\end{enumerate}
\end{algorithm-description}

\paragraph{Analysis.}
We now finish the proof by establishing the three properties of the randomized decision tree that
are stated above. The first two properties follow immediately from the description of the algorithm,
and it thus remains to establish the third property.
Fix some $x\in X_n$, and let $d\in\mathbb{N}$ be such that $\mu(x)=2^{-d}$.
We need to show that the expected number of questions that are asked when the \unknown element is $x$
is at most $d + r + r^2$. 

Let $q=q^{(i)}$ denote the sub-probability of $x$ after the $i$'th question; note that $q\in\{2^{-j} : j\leq d\}$.
\begin{lemma}\label{l721}
If $q \geq 2^{-k}$ then $q$ doubles (that is, $q^{(i+1)} = 2q^{(i)}$).
Otherwise, the expected number of questions
until $q$ doubles is at most
$\frac{1}{1-4q}$.
\end{lemma}
\begin{proof}
From the description of the algorithm, it is clear that the only case in which
the sub-probability of $x$ is not doubled is when $x$ is one of the two boundary elements in Step~4.
This only happens when $x$ is a light element (i.e.\ $q < 2^{-k}$).
The probability that $x$ is one of the boundary elements 
is at most $2q/\sigma\leq 4q$, where $\sigma\geq1/2$ is the total mass of light elements:
indeed, the probability that a given endpoint of the \halfcircle lies inside the arc corresponding to $q$ is $q/\sigma$, since each endpoint is distributed uniformly on the circle of circumference $\sigma$.

It follows that the distribution of the number of questions that pass until $q$
doubles is dominated by the geometric distribution with failure probability $4q$, and so the expected number of questions until $q$ doubles is at most $\frac{1}{1-4q}$.
\end{proof}

The desired bound on the expected number of questions needed to find $x$ follows from Lemma~\ref{l721}:
as long as $q$, the sub-probability associated with $x$, is smaller than $2^{-k}$, it takes
an expected number of $\frac{1}{1-4q}$ questions until it doubles.
Once $q\geq 2^{-k}$, it doubles after every question. Thus, by linearity of expectation,
the expected total number of questions is at most:
	\begin{align*}
		k + \sum_{j=k+1}^d \frac{1}{1-4\cdot2^{-j}}  
		&< 
		k + \sum_{j=k+1}^d [{1+ 4\cdot 2^{-j} +  2(4\cdot 2^{-j})^2}] \\
		& = d + \sum_{j=k+1}^d [{4\cdot 2^{-j} +  2(4\cdot 2^{-j})^2}] \\
		&< d + 4\cdot 2^{-k} + \frac{2}{3} (4\cdot 2^{-k})^2\\
		&< \log \frac{1}{\mu(x)} + r + r^2.
	\end{align*}
	
\section{Open questions} \label{sec:open-questions}

Our work suggests many open questions, some of which are:

\begin{enumerate}
 \item The main results of Section~\ref{sec:huffman} show that when $n = 5\cdot 2^m$, $\uhuf(n,0) = 1.25^{n \pm o(n)}$. We conjecture that there exists a function $G\colon [1,2] \to \mathbb{R}$ such that for $n = \alpha 2^m$, $\uhuf(n,0) = G(\alpha)^{n \pm o(n)}$. Our results show that $1.232 \leq G(\alpha) \leq 1.25$ and that $G(1.25) = 1.25$. What is the function $G$?
 \item Theorem~\ref{thm:minimum-redundancy-ub} constructs an optimal set of questions of size $1.25^{n + o(n)}$, but this set is not explicit. In contrast, Theorem~\ref{thm:cone} constructs explicitly an optimal set of questions of size $O(\sqrt{2}^n)$, which furthermore supports efficient indexing and efficient construction of optimal strategies. Can we construct such an explicit set of optimal size $1.25^{n + o(n)}$?
 \item The results of Section~\ref{sec:comparison-equality} show that $n \leq \uent(n,1) \leq 2n-3$. We conjecture that the limit $\beta = \lim_{n\to\infty} \frac{\uent(n,1)}{n}$ exists. What is the value of $\beta$?
\end{enumerate}

An interesting suggestion for future research is to generalize the entire theory to $d$-way questions.

\bibliographystyle{plain}
\bibliography{entropy}

\end{document}